\documentclass[reqno]{amsart}    

\usepackage[margin=1in]{geometry} 
\usepackage[parfill]{parskip}   
\usepackage{graphicx}

\usepackage{amssymb}        
\usepackage{amsmath}
\usepackage{amsthm}     
\usepackage{tikz-cd}

\usepackage{hyperref}

\usepackage{xcolor}
\usepackage{tikz}

\usepackage{booktabs}
\usepackage{tabularx}
\usepackage{array}

\usetikzlibrary{arrows.meta,positioning,calc}

\begin{document}

\newtheorem{thm}{Theorem}[section]    
\newtheorem{prop}[thm]{Proposition} 
\newtheorem{conj}[thm]{Conjecture}
\newtheorem{lem}[thm]{Lemma}
\newtheorem{cor}[thm]{Corollary}
\newtheorem{ques}[thm]{Question}

\theoremstyle{definition}
\newtheorem{definition}[thm]{Definition}
\newtheorem{example}[thm]{Example}
\newtheorem{remark}[thm]{Remark}

\def\dd{{\operatorname{d}}}

\begin{abstract}

Entropy functionals and their associated divergences underlie many statistical methods, including maximum entropy inference, minimum divergence estimation, and goodness-of-fit testing, yet choosing among Shannon, R\'enyi, Tsallis, and more general entropies is often a matter of convention rather than structural principle. We introduce a measure theoretic framework in which admissibility requires the entropy of an input measure to be bounded above by that of its reference measure whenever the former is absolutely continuous with respect to the latter. Under generalized mean-value composition, we characterize all such entropy functionals and obtain a four-level hierarchy determined successively by the mean generator, entropy scale, and additivity assumptions. A continuous strictly monotone generator $g$ is admissible exactly when $t\mapsto g(1/t)$ is strictly convex for increasing $g$, or strictly concave for decreasing $g$. This resolves a question posed by R\'enyi  (Proc.~4th Berkeley Sympos.~Math.~Statist.~Prob.,~1961) concerning which generalized means may replace the arithmetic mean in his entropy axiomatization. The same criterion is equivalent to strict convexity of an associated Csisz\'ar $f$-divergence generator and therefore yields data processing under Markov kernels with an exact equality condition. Within this hierarchy, product additivity singles out the R\'enyi family, while internal additivity, or product additivity together with arithmetic mean-value composition, singles out Shannon entropy. The characterization is constructive and yields new admissible entropy and divergence families, including integral-transform examples.
\end{abstract}

\title{A Structural Characterization of Entropy Functionals}
\author{Daniel Lazarev}
\address{Department of Mathematics, Massachusetts Institute of Technology, Cambridge, MA, USA}

\email{dlazarev@mit.edu}  

\maketitle

\section{Introduction}

Entropy functionals and the divergence families derived from them are foundational objects of statistical inference. They appear as objective functions in maximum entropy inference \cite{Jay57}, whose statistical and decision theoretic foundations were developed by Csisz\'ar \cite{Csi91} and Grünwald and Dawid \cite{GrunDaw04}, and whose roots include the
discrimination-information interpretation of Kullback and Leibler \cite{KulLeib51} and the axiomatic maximum entropy formulation of Shore and Johnson \cite{ShoJoh80}. They also appear in the broader theory of statistical divergences, including the Ali--Silvey and Csiszár $f$-divergence frameworks \cite{AliSilvey66,Csiszar67}; as the basis of the power divergence family of Cressie and Read \cite{CresRead84}; as tools for robust estimation, as in density power divergence estimation \cite{BasuHarr98}; and as estimands in their own right, as in the nonparametric estimation of Rényi entropy functionals of multidimensional densities \cite{Leo08}.

Across these settings, the same underlying question arises: among the many available entropies, what principles distinguish one from another? This question is important not only for categorizing entropy functionals, but also for understanding which statistical information measure is most appropriate in a given context, a choice that is often resolved in practice by convention or analytic convenience.

In this work, we provide a framework for entropy characterization and admissibility by first considering a minimal structural requirement: a monotonicity condition with respect to
absolute continuity, in the spirit of the information orderings underlying the comparison of statistical experiments \cite{Blackwell53, LeCam64}. 
We then characterize the entropy functionals compatible with this requirement under generalized mean-value composition, which we later show can itself be derived from more primitive structural and regularity properties. 
Thus, our results are not intended to select a single universally optimal entropy functional, but to identify the statistical assumptions underlying different entropy choices.

The characterization we obtain is constructive and answers a question posed by R\'enyi in his foundational paper \cite{Ren61} about which of the generalized (Kolmogorov-Nagumo or quasi-arithmetic) means may replace the arithmetic mean in his axiomatization of entropy and still lead to an admissible entropy, along with a stronger variant of that question in which additivity itself is relaxed. The answer takes the form of a convexity criterion on the mean's generator alone, checkable by a second-derivative test, and it induces a four-level hierarchy of entropy classes in which the familiar functionals of statistics occupy the most rigid levels. In particular, the same power-law generators that underlie the R\'enyi family and the Cressie--Read statistics are singled out by additivity across independent systems, while Shannon entropy is singled out by a stronger internal form of additivity. The hierarchy thus separates, postulate by postulate, exactly which features of the classical formulas are forced by structural monotonicity, which by mean-value composition, which by the entropy scale, and which by additivity. It thereby provides both a structural explanation for the entropy functionals in common statistical use and a principled mechanism for generating and situating new ones.

Many entropy characterizations have been proposed since Shannon's original paper \cite{Sha48}. Classical results include those of Faddeev \cite{Faddeev56} and the systematic treatment of Acz\'el and Dar\'oczy \cite{AczelDaroczy75}; see also Csisz\'ar \cite{Csi08} and Jizba and Korbel \cite{JizKor20} for broader axiomatic perspectives. More recent work has approached the problem from complementary directions. Gour and Tomamichel \cite{GourTomamichel21} study entropies and relative entropies under monotonicity with respect to mixing and data processing together with product additivity, obtaining a correspondence between the two classes. Schlather and Ditscheid \cite{SchlatherDitscheid24} decompose the classical chain rule into product additivity and an intrinsic proportionality condition. Balsubramani \cite{Balsub26} characterizes multivariate divergence functionals under data processing and product additivity. 
In contrast, product additivity is not a starting assumption here. We first classify entropy functionals under structural (reference measure) monotonicity and generalized mean-value composition, thereby retaining nonadditive admissible families, and then determine how successive scale and additivity assumptions recover the classical R\'enyi and Shannon cases.

R\'enyi's 1961 formulation \cite{Ren61} provides the conceptual starting point. To facilitate his characterization, he introduced \textit{generalized distributions}: nonnegative vectors whose total mass lies in $(0,1]$. This enlarged class is closed under both the product operation and, when the resulting mass remains at most one, disjoint union. R\'enyi used these operations, together with permutation invariance, continuity, and normalization, to characterize Shannon entropy. He then observed that there is no necessary \textit{a priori} reason for the disjoint-union law to use the arithmetic mean and asked which generalized, or quasi-arithmetic, means remain compatible with product additivity. Our finite-measure formulation removes the upper restriction on total mass, resolves the associated closure issue, and provides a natural reference measure setting in which both R\'enyi's original question and its nonadditive extension can be studied. The operations and postulates are stated formally at the beginning of Section~\ref{sec:main}.

The remainder of the paper is organized as follows. Section~\ref{sec:main} formulates R\'enyi's framework for finite reference measure spaces, introduces the new structural postulates, and states the main characterization theorems. Section~\ref{sec:stat-interpret} develops the statistical interpretation of the hierarchy, its connection with $f$-divergences and data processing, and constructions of new admissible entropy functionals. Section~\ref{sec:ent-mean-lem} establishes the required results about entropies and generalized means. Section~\ref{sec:proofs} contains the proofs of the main results and related consequences.

\section{Questions, New Postulates, and Main Results}
\label{sec:main}

\subsection{R\'enyi's postulates}

In his 1961 paper \cite{Ren61}, R\'enyi presented a natural characterization of Shannon entropy by introducing  so-called \textit{generalized distributions}, which are  distributions $P = (p_1, \dots, p_n)$, where $p_i \geq 0$ for $i=1, \dots, n$, whose mass $W(P) := \sum_{i=1}^n p_i$ satisfies $0 < W(P) \leq 1.$ Of course, if $W(P)=1$ then $P$ is simply a probability distribution.

In this paper, we write the entropies in measure-theoretic form, which allows
greater generality and gives a natural formulation of R\'enyi's operations.
In this language, the operations used in R\'enyi's characterization are, for two  measure spaces $(X, \sigma(X), \mu)$ and $(Y,\sigma(Y), \nu)$, where $\mu$ and $\nu$ are finite measures:
\begin{itemize}
    \item (Product) $(X, \mu) \otimes (Y, \nu) := (X \times Y, \mu \otimes \nu)$, where for all $A \in \sigma(X)$ and $B \in \sigma(Y)$, 
    $$
    \mu\otimes \nu(A \times B) = \mu(A) \,\nu(B) .
    $$
    \item (Disjoint sum) ~~ $(X,\mu) \oplus (Y,\nu):= (X \sqcup Y, \mu \oplus \nu)$, where for all $A \in \sigma(X)$ and $B \in \sigma(Y)$, 
    $$
    (\mu \oplus \nu)(A \sqcup B) = \mu(A) + \nu(B).
    $$
\end{itemize}
For reference, in Rényi's paper, for $Q = (q_1, \dots, q_m)$, the product is given by
$P * Q := (p_1 q_1, \dots, p_i q_k,$ $\dots, p_n q_m)$,
and the disjoint sum by
$P \cup Q := (p_1, \dots, p_n, q_1, \dots, q_m)$.
Thus, $W(P*Q) = W(P) W(Q)$ and $W(P \cup Q) = W(P)+W(Q)$, where the latter is only defined if $W(P \cup Q) \leq 1$.
More specifically, using generalized distributions allowed R\'enyi to define the disjoint sum  such that the family of generalized distributions is closed under disjoint union, which would not have been true for probability distributions. In this work, the finite-measure formulation avoids the closure issue that motivated R\'enyi's generalized distributions.

Then, for an entropy, $\mathcal{H}(\mu)$, Rényi's postulates are, in our notation, as follows. 
We note that, in Postulate 2,  we write $\{p\}$ for the finite measure on a one-point space whose unique atom has mass $p$.
\begin{itemize}
\item[] \textsc{Postulate 1.} (Relabeling Invariance) \quad
Suppose that $T:\mathcal X \rightarrow\mathcal Y$ is a measurable isomorphism satisfying $T_{\#}\mu=\mu'$ and  $T_{\#}\nu=\nu'$.
Then $
\mathcal H(\mu,(\mathcal X,\nu))
=
\mathcal H(\mu',(\mathcal Y,\nu'))$.

    \item[] \textsc{Postulate 2.} (Atomic Continuity) \quad
Let $\delta$ denote the unit measure on a one-point space. The function $p\mapsto
\mathcal H(p\delta,(\{\ast\},p\delta))$
is continuous for $p>0$.

   \item[] \textsc{Postulate 3.} (Normalization) \quad
Let $u_2$ denote the uniform probability measure on a two-point space,
and let $\#_2$ denote counting measure on that space. Then
$\mathcal H(u_2,(\{1,2\},\#_2))
=
\log 2$.
    
    \item[] \textsc{Postulate 4.} (Additivity) \quad $\mathcal{H}(\mu \otimes \nu) = \mathcal{H}(\mu) + \mathcal{H}(\nu)$.
    \item[] \textsc{Postulate 5.} (Mean Value) \quad $\mathcal{H}(\mu \oplus \nu) = \dfrac{\mu(\mathcal{X}) \mathcal{H}(\mu) + \nu(\mathcal{Y}) \mathcal{H}(\nu)}{\mu(\mathcal{X})+ \nu(\mathcal{Y})}$.
\end{itemize}

Having defined his postulates, Rényi provided a simple proof that the only entropy $\mathcal{H}$ that satisfies Postulates 1-5 is Shannon entropy. More importantly, however, he observed that there is no necessary \textit{a priori} reason to use the arithmetic mean in Postulate 5, prompting the developments and questions that we introduce below.

\subsection{Definitions and questions}
We note that means and entropy functionals are considered on their natural domains in the sense that displayed integrals are well defined and the resulting values lie in the relevant domains of the functions involved.
The general notion of a mean originating from Kolmogorov's and Nagumo's work \cite{Kol30,Nagumo1930} (see also  \cite{Aczel48, HarLitPol52, Bullen2003, deCar16})  is naturally related to the characterization of entropy. This is given by
\begin{definition}
\label{def:g-mean}
The {\it generalized (Kolmogorov-Nagumo) $g$-mean} given a measure $\mu$ on a space $\mathcal{X}$ of a measurable function $f$ is given by 
\begin{equation}
 \label{eq:mean-g}
\mathcal{M}_g (f, (\mathcal{X},\mu)) := g^{-1}\left(   \dfrac{1}{\mu(\mathcal{X})} \int_\mathcal{X}  g(f) \, \dd \mu \right) \, ,
\end{equation}
where $g$  is a continuous, strictly monotone function, with inverse $g^{-1}$.
\end{definition}

In particular, we obtain the arithmetic mean for any linear function $g(x) =ax+b$, $a,b \in \mathbb{R}$. Moreover, choosing $g(x) = 1/x$ gives the \textit{harmonic mean}, while $g(x) = \log x$ leads to the \textit{geometric mean} (where $x>0$ for both cases). Setting $g(x) = x^p$ in the generalized mean \eqref{def:g-mean} we have
\begin{definition}
\label{def:Lp-mean}
The {\it $L^p$, or power, mean} of order $p \in \mathbb{R}\setminus \{0\}$ given a finite measure $\mu$ on a space $\mathcal{X}$ of a measurable function $f$ is given by 
\begin{equation}
\label{eq:Lp-mean}
M^p_\mu (f, \mathcal{X}) := \left(   \dfrac{1}{\mu(\mathcal{X})} \int_{\mathcal{X}}  f^p \, \dd \mu \right)^{\frac{1}{p}} \, 
\end{equation}
whenever the integral exists. 
If $p<0$, we assume $f>0$ $\mu$-a.e.
The case $p=0$ is defined by the limit $p \rightarrow 0$, which gives the geometric mean in Lemma~\ref{lem:M0}.
\end{definition}
Thus, the power mean includes all three classic means mentioned above, where the other two are obtained by setting $p=1$ or $p=-1$.
When the space $\mathcal{X}$ is clear from context we will simplify the notation and write $M^p_\mu (f)$. 

Having defined the generalized mean, we can present, in our notation, Rényi's
\begin{itemize}
    \item[] \textsc{Postulate 5'.} (Generalized Mean Value) \quad There exists a continuous, strictly monotone function, $\varphi$, with inverse $\varphi^{-1}$, such that 
    $$
    \mathcal{H}(\mu \oplus \nu) = \varphi^{-1} \left(\dfrac{\mu(\mathcal{X}) \,\varphi(\mathcal{H}(\mu)) + \nu(\mathcal{Y}) \,\varphi(\mathcal{H}(\nu))}{\mu(\mathcal{X})+ \nu(\mathcal{Y})} \right).
    $$
\end{itemize}

Based on this generalized mean postulate, Rényi posed the following open problem \cite{Ren61}:
\begin{ques}
\label{renyi-ques}
    \rm What choices of the function $\varphi$ in Postulate 5' are admissible, in the sense of being compatible with Postulate 4?
\end{ques}

Of course, the question itself suggests that additivity (Postulate 4) is essential in defining entropy. Nonadditive entropy functionals, however, also arise naturally. For example, Perelman's entropy for Ricci flow and the Colding--Minicozzi entropy for mean-curvature flow are used through their monotonicity properties \cite{Per02,ColMin12}; Tsallis entropy serves as an information-theoretic example below (Example~\ref{ex:tsallis}). This motivates replacing Postulate~4 by the weaker 
\begin{itemize}
    \item[] \textsc{Postulate A.} (Structural Monotonicity) \quad For every pair of nonzero finite measures $\mu \ll \nu$ on $\mathcal{X}$,
    $$
    \mathcal{H}(\mu, (\mathcal{X}, \nu)) \leq \mathcal{H}(\nu, (\mathcal{X}, \nu)),
    $$ 
    with equality if and only if $\mu = c \nu$ for some constant $c>0$. 
    Moreover, for any two finite measure spaces $(\mathcal{X}, \nu)$ and $(\mathcal{Y}, \xi)$, if $\nu(\mathcal{X}) = \xi(\mathcal{Y})>0$, then 
    $$\mathcal{H}(\nu, (\mathcal{X}, \nu)) = \mathcal{H}(\xi, (\mathcal{Y}, \xi)).
    $$
\end{itemize}

It will be convenient to write explicitly the atomic entropy scale already
induced by Postulate~2:
\begin{equation}
\label{eq:atomic-h}
h(r)
:=
\mathcal H
\left(
r^{-1}\delta,
(\{\ast\},r^{-1}\delta)
\right),
\qquad r>0.
\end{equation}
Thus $h(r)$ is the entropy of a one-atom reference structure of total
mass $1/r$.

\begin{remark}
\label{rem:ent-notation}
Throughout the postulates, the notation $\mathcal{H}(\mu)$ suppresses the ambient
reference structure when that structure is clear from context. Thus, when
$\mu$ is considered on the reference measure space
$(\mathcal{X},\nu)$, the notation $\mathcal{H}(\mu)$ means
$\mathcal{H}(\mu,(\mathcal{X},\nu))$, as we used explicitly in Postulate A.
This convention is important: if $\mu\ll\nu$ and $\nu\ll\mu$, the two relative entropies $\mathcal{H}(\mu,(\mathcal{X},\nu))$ and $\mathcal{H}(\nu,(\mathcal{X},\mu))$ are defined with respect to different reference structures and need not agree.

The same convention is used for products and disjoint sums, whose reference measures are the corresponding product and disjoint-sum reference measures. 
If, for $i=1,2$, the measure $\mu_i$ is evaluated relative to the reference
measure space $(\mathcal{X}_i,\nu_i)$, and if $m_i:=\mu_i(\mathcal{X}_i)$ and $m:=m_1+m_2$,
then $\mathcal{H}(\mu_1\oplus\mu_2)$ is understood to be evaluated relative to
$$
\left(
\mathcal{X}_1\sqcup\mathcal{X}_2,
\frac{m_1}{m}\nu_1\oplus\frac{m_2}{m}\nu_2
\right).
$$
With this convention, on the component $\mathcal{X}_i$ one has
$$
\frac{1}{m}
\frac{\dd(\mu_1\oplus\mu_2)}
{\dd\left(\frac{m_1}{m}\nu_1\oplus\frac{m_2}{m}\nu_2\right)}
=
\frac{1}{m_i}\frac{\dd\mu_i}{\dd\nu_i}.
$$
Thus disjoint union preserves the normalized Radon--Nikodym density of each
summand, while the masses $m_i$ determine the weights in Postulates 5 and $5'$.
\end{remark}
In the case of information theoretic entropy and Rényi's axiomatization, this monotonicity is also hidden in the generalized mean value property of Postulate 5'. It justifies, from an entropic perspective,  why $\varphi$ in Postulate 5' must be strictly monotone: Otherwise, the conjugated application of $\varphi$ to the disjoint sum of atomic entropies would not guarantee that the resulting entropy of the disjoint sum remains monotonic.  

\begin{remark}
\label{rem:post5'}
Following R\'enyi's example, we can question the need for Postulate 5' itself. Indeed, as we show in Proposition~\ref{prop:derive-5'}, Postulate 5' may itself be derived from several of the entropy postulates together with intensivity, structural compatibility, and standard regularity assumptions on the disjoint sum operation.
Therefore, in short, statistical entropy can be described as an intensive monotonic coordinate on a family of finite measure spaces.
Having noted this, we will make use of Postulate 5' for reasons of clarity, and to better situate our results relative to R\'enyi's original work.
\end{remark}

The above discussion leads to the following modified version of Rényi's problem:
\begin{ques}
\label{renyi-ques'}
    \rm What choices of the function $\varphi$ in Postulate 5' are admissible, in the sense of being compatible with Postulate A?
\end{ques}

Our main results, Theorems \ref{thm:ent-admit} and \ref{thm:ent-char}, address both Questions \ref{renyi-ques} and \ref{renyi-ques'}.
We start by defining the generalized entropy.

\begin{definition}
\label{def:gh-ent}
The {\it generalized $(g,h)$-entropy} given a measure $\mu$ on a measure space $(\mathcal{X}, \nu)$, with $\mu \ll \nu$,  is given by 
\begin{equation}
 \label{eq:gh-entropy}
\mathcal{H}_{g,h}(\mu, (\mathcal{X}, \nu)) =  h \left( \mathcal{M}_g \left(\dfrac{1}{\mu(\mathcal{X}) } \dfrac{\dd \mu}{\dd \nu}, (\mathcal{X}, \mu) \right) \right)  
=   h \circ g^{-1}\left(   \dfrac{1}{\mu(\mathcal{X})} \int_\mathcal{X}  g  \left(\dfrac{1}{\mu(\mathcal{X})}\dfrac{\dd \mu}{\dd \nu} \right) \ \, \dd \mu \right) , 
\end{equation}
where $h$ is a continuous strictly decreasing function, and $g$  is a continuous, strictly monotone function, with inverse $g^{-1}$.
\end{definition}

\begin{remark}
The function $h\left(\dfrac{1}{\mu(\mathcal{X}) } \dfrac{\dd \mu}{\dd \nu}\right)$ is the ``atomic entropy'' (Eq.~\ref{eq:atomic-h}) with $h$ serving as a choice of scale,  much like the normalization of Postulate 3 can be viewed as a choice of units.
In fact, as we will show explicitly in the proof of Theorem~\ref{thm:ent-char}  in Sec.~\ref{sec:proofs}, the entropy that emerges directly from the weakest postulates we presented, in terms of the same strictly monotone map $\varphi$ of Postulate 5', is given by
\begin{equation}
\label{eq:M-h-inside}
\mathcal{H}_{\varphi,h}(\mu, (\mathcal{X}, \nu)) =  \mathcal{M}_\varphi \left(h\left(\dfrac{1}{\mu(\mathcal{X}) } \dfrac{\dd \mu}{\dd \nu}\right), (\mathcal{X}, \mu) \right)  
=
 \varphi^{-1}\left(   \dfrac{1}{\mu(\mathcal{X})} \int_\mathcal{X}  \varphi \circ h  \left(\dfrac{1}{\mu(\mathcal{X})}\dfrac{\dd \mu}{\dd \nu} \right) \ \, \dd \mu \right)
\end{equation}
where the role of $h$ as an atomic entropy is more clear. We derive Eq.~\eqref{eq:gh-entropy} from Eq.~\eqref{eq:M-h-inside} by taking  $g = \varphi \circ h$.
Nevertheless, the form presented in Definition~\ref{def:gh-ent} will be more useful in what follows.
\end{remark}

\begin{remark}
    A discrete version of Eq.~\eqref{eq:M-h-inside} with $h(r) = - \log r$ appeared in Acz{\'e}l and Dar{\'o}czy's classic book (Eq.~2.2.11 in \cite{AczelDaroczy75}). Definition~\ref{def:gh-ent} extends this form to finite measure spaces equipped with reference measures and separates the mean generator, $g$, from the entropy scale, $h$. Within the present hierarchy, Theorem~\ref{thm:ent-char} shows that the resulting $(g,h)$-entropies form the most general class compatible with structural monotonicity and generalized mean-value composition. Related entropy functionals have also been studied under the name ``generalized entropies'' \cite{GrunDaw04, DupKra12, AmiBalHer18, BalPal20}.
\end{remark}

At the most basic level, Postulate 4 is required to determine a choice of scale, $h$, while Postulate 5 or 5' constrain the possible generator, $g$. However, as stated, Postulate 4 does more than that: it does not independently choose a particular $h$ but also restricts the possible $g$. Therefore, postulates that only delimit the atomic scale $h$, such as Postulate $4_{\rm at}$ which we now introduce, represent the weakest form of additivity, in the sense of decoupling the constraints on $h$ from $g$. 
\begin{itemize}
    \item[] \textsc{Postulate $4_{\rm at}$.} (Atomic Additivity) \quad
    Let $h:(0,\infty)\to\mathbb{R}$ be defined as in Eq.~\ref{eq:atomic-h}. Then, for all $r,s>0$, \; $h(rs)=h(r)+h(s)$.
\end{itemize}

\begin{remark}
\label{rem:h-log}
Together with continuity, Postulate $4_{\rm at}$ implies that
$h(r)=a\log r $
for some constant $a$. The direction of Postulate A requires the atomic
entropy scale to be decreasing, so $a<0$, and Postulate 3 gives $h(1/2)=\log 2$,
so $a=-1$. Therefore $h(r)=-\log r$.
\end{remark}
Since $h$ is a strictly decreasing function, any specific choice of $h$ is always order-equivalent to any other, so we do not lose generality, at the level of Postulate A, by choosing a particular coordinate. 
In fact, the admissibility criterion of Theorem~\ref{thm:ent-admit} depends only on the induced mean generator $g$, and not separately on the inner entropy scale $h$ (though, of course, this does not mean that the criterion is independent of the entropy scale when the
generator $\varphi$ in Postulate $5'$ is held fixed, since
$g=\varphi\circ h$).
As detailed in Example~\ref{ex:tsallis}, this order-equivalence has already been noticed by Leonenko, Pronzato, and Savani \cite{Leo08}, who point out that R\'enyi and Tsallis entropies are functions of each other and are maximized by the same distributions, prompting them to refer to ``$q$-entropy maximizing distributions,'' using the exponent $q$ without needing to mention if this distribution was estimated using the corresponding R\'enyi or Tsallis entropies.

Postulate $4_{\rm at}$ and Remark~\ref{rem:h-log}  lead to the following 
\begin{definition}
\label{def:g-ent}
The {\it generalized $g$-entropy} given a measure $\mu$ on a measure space $(\mathcal{X}, \nu)$, with $\mu \ll \nu$,  is given by 
\begin{equation}
 \label{eq:g-entropy}
\mathcal{H}_g(\mu, (\mathcal{X}, \nu)) = - \log \mathcal{M}_g \left(\dfrac{1}{\mu(\mathcal{X}) } \dfrac{\dd \mu}{\dd \nu}, (\mathcal{X}, \mu) \right) 
=  - \log g^{-1}\left(   \dfrac{1}{\mu(\mathcal{X})} \int_\mathcal{X}  g\left( \dfrac{1}{\mu(\mathcal{X})}\dfrac{\dd \mu}{\dd \nu} \right) \, \dd \mu \right) \, ,
\end{equation}
where $g$  is a continuous, strictly monotone function, with inverse $g^{-1}$.
\end{definition}

Having defined the generalized entropy, we can now also distinguish between two forms of additivity.  R\'enyi's Postulate 4 is an \textit{external} additivity condition in that it concerns the product of two measure spaces. Using the notation established in Remark~\ref{rem:ent-notation} and Definitions~\ref{def:gh-ent} and \ref{def:g-ent}, we can now more clearly rewrite 

\begin{itemize}
    \item[] \textsc{Postulate 4.} (External Additivity) \quad For $i=1,2$, let $\mu_i$ be finite measures on the measure spaces $(\mathcal{X}_i, \nu_i)$, where $\mu_i\ll\nu_i$. Then
    $$
    \mathcal{H}_g\left(
    \mu_1\otimes\mu_2,
    (\mathcal X_1\times\mathcal X_2,\nu_1\otimes\nu_2)
    \right)
    =
    \mathcal{H}_g(\mu_1,(\mathcal X_1,\nu_1))
    +
    \mathcal{H}_g(\mu_2,(\mathcal X_2,\nu_2)).
    $$
\end{itemize}

There is also a stronger form of additivity, which acts within a single measure
space rather than between two measure spaces.  To state it cleanly,
fix a finite measure space $(\mathcal X,\mu)$ and define, for a positive
measurable function $f$,
$$
\widetilde{\mathcal{H}}_g(f,(\mathcal X,\mu))
:=
-\log \mathcal{M}_g(f,(\mathcal X,\mu)).
$$

\begin{itemize}
    \item[] \textsc{Postulate $4_{\mathrm{int}}$.} (Internal Additivity) \quad
    For all positive measurable functions $f,f'$ on a finite measure space
    $(\mathcal X,\mu)$,
    $$
    \widetilde{\mathcal{H}}_g(ff',(\mathcal X,\mu))
    =
    \widetilde{\mathcal{H}}_g(f,(\mathcal X,\mu))
    +
    \widetilde{\mathcal{H}}_g(f',(\mathcal X,\mu)).
    $$
\end{itemize}
Of course, 
$$
\mathcal{H}_g\big(\mu,(\mathcal X,\nu) \big)
=
\widetilde{\mathcal{H}}_g \left(\frac{1}{\mu(\mathcal X)}
\frac{\dd\mu}{\dd\nu},(\mathcal X,\mu) \right).
$$

Put succinctly, Postulate 4 is ``external'' because it combines independent systems  ($f\otimes f'$ on $\mathcal{X}\times\mathcal Y$), while Postulate $4_{\mathrm{int}}$ is ``internal'' because it combines densities inside one fixed system ($ff'$ on $\mathcal{X}$). Lemma~\ref{lem:4int-4} in Sec.~\ref{sec:ent-mean-lem} confirms that Postulate $4_{\mathrm{int}}$ indeed implies Postulate 4.

We proceed by defining the Shannon and R\'enyi entropies:
\begin{definition}
\label{def:shannon}
The \textit{Shannon entropy} 
of a  finite  measure $\mu$  on a measure space $(\mathcal{X}, \nu)$, with $\mu \ll \nu$, is given by
\begin{equation}
 \label{eq:ent-def}
S_\nu(\mu,\mathcal{X}) 
= \log \mu(\mathcal{X}) 
 - \frac{1}{\mu({\mathcal{X}})} \int_{{\mathcal{X}}}   \dfrac{\dd \mu}{\dd \nu} 
 \log \left(\dfrac{\dd\mu}{\dd \nu} \right)  \dd \nu  .
\end{equation}
\end{definition}

\begin{definition}
\label{def:renyi-ent}
Using the same notation as in Definition~\ref{def:shannon}, the \textit{R\'enyi entropy} of order $p \in \mathbb{R}_{\geq 0} \setminus \{1\}$ is given by
\begin{equation}
 \label{eq:renyi-def}
H^p_\nu(\mu,\mathcal{X}) 
=  \log \mu(\mathcal{X}) + \dfrac{1}{1-p }  \log \left( 
 \frac{1}{\mu({\mathcal{X}})} \int_{{\mathcal{X}}} \left(\dfrac{\dd \mu}{\dd \nu} \right)^p
 \dd \nu \right) .
\end{equation}
We define the R\'enyi entropy for $p=1$ by taking the limit $p \rightarrow 1$, in which case, by Lemma~\ref{lem:M0} and Remark~\ref{rem:ren-ent-M}, it equals Shannon entropy. 
For $p=0$, we use the standard convention,
$\left(\frac{\dd\mu}{\dd\nu}\right)^0
:=
\mathbf 1_{\{\dd\mu/\dd\nu>0\}}$,
and set
$$
H_\nu^0(\mu,\mathcal X) :=\log\nu(\{\dd\mu/\dd\nu>0\}).
$$
\end{definition}

\begin{remark}
    The Shannon entropy \eqref{eq:ent-def} follows from the generalized entropy \eqref{eq:g-entropy} by taking $g(x) = \log (x)$.
    The R\'enyi entropy \eqref{eq:renyi-def} follows from the generalized entropy \eqref{eq:g-entropy} by taking $g(x) = x^{p-1}$, with $p$ as in Definition~\ref{def:renyi-ent}.
\end{remark}

\begin{remark}
The above entropies are usually defined for input measures that are probability measures, and, as expected, taking $\mu(\mathcal{X}) = 1$  reduces the above formulas to the more common expressions. However, expressing the entropies more generally for finite measures will be more useful in what follows. We note that Shannon entropy as written in Definition~\ref{def:shannon} is also known as ``relative entropy'' or ``Baron-Jauch entropy'' \cite{BarJau72, Whe91}, although these are usually written for probability measures only, not in the finite measure form presented above, which was also previously used in \cite{Laz23}. 
\end{remark}

\subsection{Main results}

\begin{thm}[Entropy Admissibility]
    \label{thm:ent-admit}
Let $\mathcal{H}_{g,h}$ be a generalized
$(g,h)$-entropy of the form \eqref{eq:gh-entropy}, with $h$ continuous  and strictly decreasing,
and $g$ continuous and strictly monotone. Then $\mathcal{H}_{g,h}$ satisfies Postulate~A
on the class of finite measure spaces if and only if the following equivalent
conditions hold:
\begin{enumerate}
  \item[(a)] For every finite measure space $(\mathcal{X},\mu)$ and every measurable function
$f>0$ $\mu$-a.e. for which the displayed means are defined,
  \begin{equation}
\label{eq:M-condition}
\mathcal{M}_g(f,(\mathcal{X},\mu))\geq M^{-1}_\mu(f,\mathcal{X}),
\end{equation}
with equality if and only if $f$ is constant $\mu$-a.e.
  \item[(b)] if $g$ is increasing (decreasing) then $g(1/t)$ is strictly convex (concave) for $t \in(0,\infty)$. 
\end{enumerate}
\end{thm}


The next theorem uses the following regularity condition.

\begin{definition}
\label{def:simple-density-continuity}
Suppose that $\mathcal H$ satisfies Postulate $5'$ with generator
$\varphi$, let $h$ be its atomic entropy scale, and set $g:=\varphi\circ h$.
For a nonzero finite measure $\mu$ and a positive measurable function
$f$, let $\nu_f$ denote the reference measure defined by
$$
\dd\nu_f
:=
\frac{1}{\mu(\mathcal X)f}\,\dd\mu,
$$
whenever this measure is finite. We say that $\mathcal H$ is
\textit{continuous under simple density approximation} if, for every
sequence of positive simple functions $f_n$,
$$
g(f_n)\longrightarrow g(f)
\quad\text{in }
L^1\left(\frac{\mu}{\mu(\mathcal X)}\right)
$$
implies
$$
\mathcal H(\mu,(\mathcal X,\nu_{f_n}))
\longrightarrow
\mathcal H(\mu,(\mathcal X,\nu_f)).
$$
\end{definition}

\begin{thm}[Entropy Characterization]
\label{thm:ent-char}
Let $\mathcal H$ be an entropy on finite reference-measure spaces
satisfying Postulates $A$ and $1$--$3$. 
Let $h$ be its induced atomic entropy scale. Assume also
continuity under simple density approximation in the sense of
Definition~\ref{def:simple-density-continuity}.

Assume further that, for every nonzero finite measure $\mu\ll\nu$, if
$$
Y
:=
\left\{
x\in\mathcal X:
\frac{\dd\mu}{\dd\nu}(x)>0
\right\},
$$
then
$$
\mathcal H(\mu,(\mathcal X,\nu))
=
\mathcal H
\bigl(
\mu|_Y,(Y,\nu|_Y)
\bigr).
$$
Then, on the class of reference-measure structures satisfying
$\mu\ll\nu$, and in order of decreasing generality:
\begin{enumerate}
\item If $\mathcal H$ satisfies Postulate $5'$, then it is given by the
$(g,h)$-entropy \eqref{eq:gh-entropy}, with
$$
g=\varphi\circ h,
$$
where $\varphi$ is the generator in Postulate $5'$.

\item If $\mathcal{H}$ satisfies Postulates $5'$ and $4_{\rm at}$, then it is
given by the $g$-entropy \eqref{eq:g-entropy}.

\item If $\mathcal H$ satisfies Postulates $5'$ and $4$, then it is given by
R\'enyi entropy \eqref{eq:renyi-def}, with $p>0$.

\item If the resulting $g$-entropy satisfies Postulate $4_{\rm int}$, or if
$\mathcal H$ satisfies Postulates $4$ and $5$, then $\mathcal H$ is Shannon
entropy \eqref{eq:ent-def}.
\end{enumerate}
\end{thm}

\begin{remark}
The same admissibility condition also has an equivalent statistical
formulation as data processing for the associated Csisz\'ar $f$-divergence, or as entropy monotonicity under Markov kernels; see Proposition~\ref{prop:data-processing}.
In Theorem~\ref{thm:ent-admit}, condition $(b)$ sharpens condition $(a)$
into a criterion that depends only on $g$, independent of $\mu$, and $\nu$. For Shannon
entropy, $g=\log$ is increasing and $g(1/t)=-\log t$ is strictly convex on $(0,\infty)$. For R\'enyi entropy of order $p$, $g(t)=t^{p-1}$;
for $p>1$, $g$ is increasing and $g(1/t)=t^{1-p}$ is strictly convex,
while for $0<p<1$, $g$ is decreasing and $g(1/t)=t^{1-p}$ is strictly
concave. Thus, Theorem~\ref{thm:ent-admit} answers Question~\ref{renyi-ques'}
explicitly:
$$
   g \text{ is admissible in Postulate~5}'
  \iff
    g \text{ increasing (decreasing) and } t\mapsto g(1/t) \text{ strictly convex (concave)}.
$$
\end{remark}

\begin{remark}
\label{rem:summary}
    We may summarize the hierarchy of the four regimes in Theorem~\ref{thm:ent-char} based on the postulates they satisfy:
$$
\begin{array}{ccccccc}
\begin{array}[t]{c}
4+5\\
\text{or } 4_{\mathrm{int}}+5'
\end{array}
& \subseteq &
4+5'
& \subseteq &
A+5'+4_{\rm at}
& \subseteq &
A+5'
\\[6pt]
\text{(Shannon entropy)}
&
&
\text{(R\'enyi entropy)}
&
&
\text{($g$-entropy)}
&
&
\text{($(g,h)$-entropy)} .
\end{array}
$$
\end{remark}

Fig.~\ref{fig:postulate-tree} summarizes the relationship between the various postulates and the corresponding entropies.

\begin{figure}[t]
\centering
\resizebox{0.92\textwidth}{!}{%
\begin{tikzpicture}[
    >=Latex,
    every node/.style={align=center, font=\small},
    post/.style={
        draw,
        rounded corners,
        thick,
        inner sep=3pt,
        fill=blue!3,
        text width=2.8cm
    },
    outcome/.style={
        draw,
        rounded corners,
        thick,
        inner sep=3pt,
        fill=green!5,
        text width=2.7cm
    },
    aux/.style={
        draw,
        rounded corners,
        thick,
        inner sep=3pt,
        fill=orange!5,
        text width=2.9cm
    },
    line/.style={->, thick}
]

\node[aux, text width=4.1cm] at (0,0) (A5p)
{Postulate A\\
(monotonicity)\\
Postulate $5'$\\
(quasi-mean law)};

\node[outcome] at (5.2,0) (ghent)
{generalized\\
$(g,h)$-entropy};

\node[aux] at (-4.2,-3.0) (p5)
{Postulate $5$\\
(arithmetic mean)};

\node[post] at (0,-3.0) (p4at)
{Postulate $4_{\rm at}$\\
(atomic additivity)\\
$h=-\log$};

\node[post] at (4.2,-3.0) (p4q)
{Postulate $4_q$\\
($q$-additivity)\\
$h=h_q$};

\node[outcome] at (2.0,-5.0) (gent)
{generalized\\
$g$-entropy};

\node[outcome] at (4.2,-6.2) (tsallis)
{Tsallis\\
entropy};

\node[post] at (0,-6.2) (p4)
{Postulate $4$\\
(external \\additivity)};

\node[aux, text width=1.8cm] at (-3.3,-8.9) (p45)
{Postulates $4+5$};

\node[post] at (0,-8.9) (p4int)
{Postulate $4_{\rm int}$\\
(internal \\additivity)};

\node[outcome] at (4.2,-8.9) (renyi)
{R\'enyi\\
entropy};

\node[outcome] at (-0.9,-11.6) (shannon)
{Shannon\\
entropy};

\draw[line] (A5p) -- (ghent);
\draw[line] (A5p) -- (p5);
\draw[line] (A5p) -- (p4at);
\draw[line] (A5p) -- (p4q);

\draw[line] (p4at) -- (p4);
\draw[line] (p4at.south east) to[out=-40,in=100] (gent.north);
\draw[line] (p4q) -- (tsallis);

\draw[line] (p4) -- (renyi);
\draw[line] (p4) -- (p4int);

\draw[line] (p5) -- (p45);
\draw[line] (p4) -- (p45);
\draw[line] (p45) -- (shannon);
\draw[line] (p4int) -- (shannon);

\end{tikzpicture}%
}
\caption{Relationship between the postulates and the corresponding entropy
regimes. Arrows from postulates to postulates indicate additional or
stronger assumptions, while arrows to green boxes indicate the resulting
entropy class. The branches $4_{\rm at}$ and $4_q$ represent different
choices of entropy scale. The node $4+5$ and the branch
$4_{\rm int}+5'$ record the two routes to Shannon entropy.}
\label{fig:postulate-tree}
\end{figure}
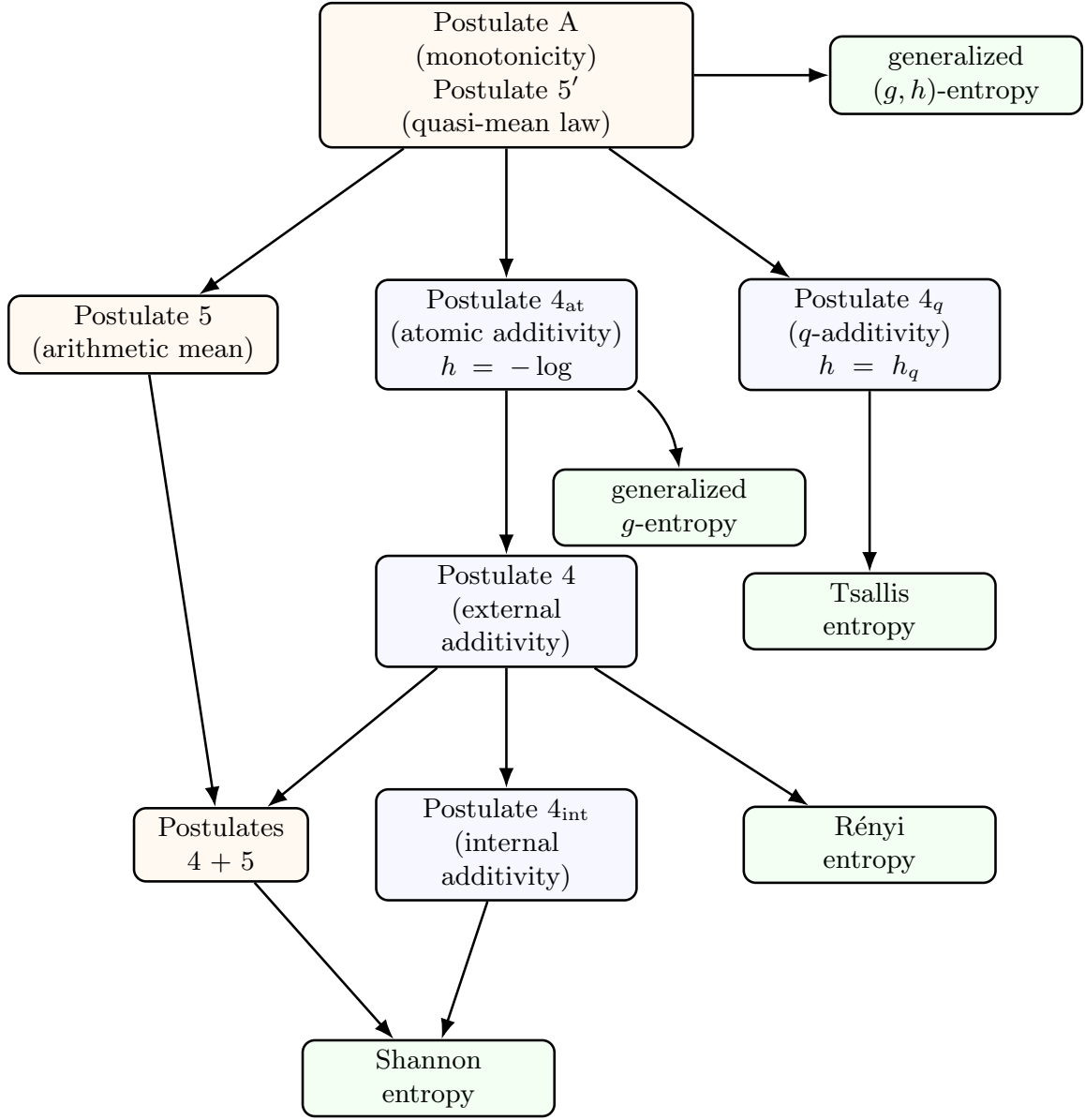

\begin{remark}
R\'enyi's formulation of Question~\ref{renyi-ques} suggests two distinct notions of admissibility. If admissibility is interpreted as compatibility with the full additivity Postulate 4, then Theorem~\ref{thm:ent-char} shows that the admissible generators reduce to the power-law and logarithmic cases, giving the R\'enyi--Shannon branch of the hierarchy. However, R\'enyi's discussion of the power law generator as representing one option among other admissible choices points to a broader admissibility problem that is only possible by weakening his additivity postulate (Postulate 4). This is precisely Question~\ref{renyi-ques'}, and Theorem~\ref{thm:ent-admit} answers it by the convexity/concavity criterion for $t\mapsto g(1/t)$. This separation is also consistent with R\'enyi's later treatment of relative information measures, for which he proved that the same postulates select the logarithmic and power-law generators. Thus Theorem~\ref{thm:ent-admit} isolates the general admissibility criterion, while Theorem~\ref{thm:ent-char} explains why the stronger additivity assumptions recover the classical Rényi and Shannon cases.
\end{remark}

We end this section by considering Tsallis entropy \cite{Tsallis88} as an example of how an entropy we have not yet considered fits within the framework presented above.
\begin{example}
\label{ex:tsallis}
     The \textit{Tsallis entropy} is defined in our notation as 
    $$
S^q_\nu(\mu,\mathcal{X}) 
=   \dfrac{1}{q-1 } \left(  1-
 \frac{1}{\mu({\mathcal{X}})}\int_{{\mathcal{X}}} \left( \frac{1}{\mu({\mathcal{X}})}\dfrac{\dd \mu}{\dd \nu} \right)^{q-1}
 \dd \mu  \right) 
 =   \dfrac{1}{q-1 } \left(  1-
 \frac{1}{\mu({\mathcal{X}})^q}\int_{{\mathcal{X}}} \left( \dfrac{\dd \mu}{\dd \nu} \right)^q
 \dd \nu  \right) .
$$
Equivalently, if
$$
g(t)=t^{q-1},
\qquad
h_q(u)=\frac{1-u^{q-1}}{q-1},
$$
then
$$
S^q_\nu(\mu,\mathcal X)
=
h_q\!\left(
\mathcal{M}_g\left(\frac{1}{\mu(\mathcal X)}
\frac{d\mu}{d\nu},(\mathcal X,\mu) \right)
\right).
$$
Thus, Tsallis entropy uses the same mean generator
$g(t)=t^{q-1}$ as R\'enyi entropy of order $q$, but with a different entropy scale. 
R\'enyi entropy corresponds to the scale
$h(u)=-\log u$, while Tsallis entropy corresponds to $h_q$.

This difference in the inner scale is exactly reflected in the product
composition law. The logarithmic scale satisfies ordinary atomic additivity, per Postulate $4_{\rm at}$,
$$
-\log(rs)= -\log r-\log s,
$$
and imposing the same additivity for full (non-atomic) entropies leads to the ordinary additivity of Postulate $4$. By contrast, the Tsallis scale satisfies
$$
h_q(rs)
=
h_q(r)+h_q(s)+(1-q)h_q(r)h_q(s),
$$
which leads to the corresponding product rule given by
\begin{itemize}
    \item[] \textsc{Postulate $4_{q}$.} ($q$-Additivity) \quad $\mathcal{S}(\mu \otimes \nu) = \mathcal{S}(\mu) + \mathcal{S}(\nu) +(1-q)\mathcal{S}(\mu) \mathcal{S}(\nu)$.
\end{itemize}
Nevertheless, since the admissibility condition of
Theorem~\ref{thm:ent-admit} depends only on the mean generator $g$, Tsallis
entropy satisfies Postulate A for the same reason as R\'enyi entropy of the
same order. For fixed $q$, Tsallis entropy is therefore a monotone reparameterization of R\'enyi entropy and determines the same entropy ordering, as previously noted \cite{CzaNau02, DukkBhat07, Leo08}.
However, it belongs naturally to the generalized $(g,h)$-entropy class rather than to the generalized $g$-entropy class (see Fig.~\ref{fig:postulate-tree}), because it replaces the atomic additivity of Postulate $4_{\rm at}$ by the $q$-additivity scale above. The two scales coincide in the Shannon limit $q \to 1$.
\end{example}

\vspace{0.5em}

\section{Statistical Interpretation and Applications}
\label{sec:stat-interpret}

\subsection{Interpretation and use of the hierarchy}

Theorem~\ref{thm:ent-char} separates four structural choices that are
often conflated in the definition of entropy: comparison with a reference
measure, aggregation across disjoint components, the numerical entropy
scale, and composition under products.

At the first level, Postulate A requires the proportional reference case
to be the unique entropy maximizer, while Postulate $5'$ requires
disjoint components to be combined through a generalized (quasi-arithmetic) mean. The
resulting $(g,h)$-entropy separates two roles: $g$ determines how the
normalized Radon--Nikodym density is averaged, while $h$ determines the
scale on which that mean is reported. The ordering imposed by Postulate A only constrains $h$ to be strictly decreasing, so different choices of $h$ may change its numerical scale while preserving the same entropy ordering.
Postulate $4_{\rm at}$ essentially fixes this outer scale to be given by
$h(r)=-\log r$.
This produces the generalized $g$-entropies.

External additivity (Postulate 4) imposes a stronger restriction. For product reference-measure structures, the normalized densities multiply. Requiring the corresponding generalized mean to factor across products selects the power-law generators, and hence the R\'enyi family of entropies. Thus, external additivity restricts not only the entropy scale but also the manner in which density information is averaged, restricting $g$.

Shannon entropy is obtained at the most rigid level. R\'enyi's original route combines external additivity (Postulate 4) with the arithmetic mean-value
postulate (Postulate 5). Alternatively, Postulate $4_{\rm int}$ requires factorization under multiplication of positive functions within a fixed measure space. Among the power means, only the geometric mean has this property, giving 
Shannon entropy.

The hierarchy may therefore be used diagnostically. Structural
monotonicity, as given by Postulate A, determines whether a proposed generator is admissible.
Atomic additivity fixes the logarithmic entropy scale, a common choice of scale, though other choices can be equally reasonable in certain contexts ({\it e.g.}, Tsallis entropy in Example~\ref{ex:tsallis}). External additivity reduces the class to R\'enyi entropy, while internal
additivity singles out Shannon entropy. The framework consequently
identifies which structural assumptions are retained or relinquished when passing among Shannon, R\'enyi, Tsallis, and more general admissible entropy functionals.

\subsection{Admissible entropy functionals: construction and examples}
\label{subsec:admissible-examples}

Theorem~\ref{thm:ent-admit} is constructive rather than merely
classificatory: once an admissible generator, $g$, is identified, it
immediately determines a generalized entropy through
Definition~\ref{def:gh-ent} after a choice of entropy scale, $h$. We first give a practical differential
test for admissibility and then illustrate the result with both established
entropy families and further admissible functionals generated by the
theorem. 

Throughout the remainder of this subsection, $P$ and $Q$ denote
probability measures. Whenever $dP/dQ$ appears, $P\ll Q$ is understood;
stronger equivalence or positivity assumptions are stated explicitly where
they are required.

To ease the statement of the results below we define
$$
\varepsilon_g
:=
\begin{cases}
1, & g \textrm{ increasing},\\
-1, & g \textrm{ decreasing}.
\end{cases}
$$

\begin{cor}[Differential admissibility criterion]
\label{cor:differential-admissibility}
Let $g\in C^2(0,\infty)$ be strictly monotone. Then $g$ is admissible whenever
$$
\varepsilon_g \cdot\big(2g'(x)+xg''(x) \big)>0
\qquad
\text{for all }x>0.
$$
\end{cor}

\begin{proof}
Direct differentiation gives
$$
\frac{d^2}{dt^2}g(1/t)
=
\frac{1}{t^4}
\left[
g''(1/t)+2t\,g'(1/t)
\right].
$$
Setting $x=1/t$, the sign of this expression is the sign of $2g'(x)+xg''(x)$.
The result therefore follows from Theorem~\ref{thm:ent-admit}.
\end{proof}

The admissibility criterion also has a useful connection with the classical
theory of statistical divergences. Let $P\ll Q$ be probability measures
and write $r=dP/dQ$. If $g$ is admissible, then
$$
F_g(t):=\varepsilon_g\cdot \big(t g(t)-g(1) \big), \quad t>0
$$
is strictly convex, which follows from Theorem~\ref{thm:ent-admit}. If $F_g \in C^2(0, \infty)$, then
$$
F_g''(r)= \varepsilon_g\cdot \big(2g'(r)+r g''(r)\big) >0.
$$
Moreover,
$$
D_{F_g}(P\Vert Q)
:=
\int_{\mathcal{X}}
F_g\left(\frac{dP}{dQ}\right)\,dQ
=
\varepsilon_g\cdot \left(\int_{\mathcal{X}}
g\left(\frac{dP}{dQ}\right)\,dP-g(1) \right),
$$
and hence
$$
H_g(P,(\mathcal{X},Q))
=
-\log
g^{-1}\left(
g(1)+\varepsilon_gD_{F_g}(P\Vert Q)
\right).
$$
Thus an admissible generator determines an associated Csisz\'ar
$f$-divergence, and the corresponding $g$-entropy is a monotone
transformation of that divergence. Theorem~\ref{thm:ent-admit}, however,
characterizes generalized-mean generators compatible with Postulate~A;
it is not itself a characterization of arbitrary $f$-divergences. This connection places the resulting functionals within the broader literature
on minimum-divergence estimation, density-ratio estimation, proper losses,
and variational divergence methods
\cite{Csiszar67,NguWaiJor10,ReidWill11,NowCseTom16}.

\begin{prop}[Admissible representations of $f$-divergences]
\label{prop:f-divergence-representation}
Let $F\in C^1(0,\infty)$ be strictly convex and satisfy
$F(1)=0$. For $c\in\mathbb R$ and $\varepsilon \in \{1, -1\}$, define
$$
g_{\varepsilon,c}(t):=\frac{\varepsilon F(t)+c}{t}.
$$
Then the function $\varepsilon g_{\varepsilon,c}$ is an increasing admissible generator if and only if
$$
t  F'(t) - F(t)> \varepsilon c
\qquad
\text{for every }t>0.
$$
In this case, $F(t)=\varepsilon \cdot\big(tg_{\varepsilon,c}(t)-g_{\varepsilon,c}(1) \big)$.

Consequently, $F$ admits an admissible entropy-generator
representation if and only if
\begin{equation}
\label{ineq:f-div-admit}
\inf_{t>0}
\left\{
t F'(t)-F(t)
\right\} > -\infty.
\end{equation}
\end{prop}

\begin{proof}
Set
$k_{\varepsilon,c}
:=
\varepsilon g_{\varepsilon,c}$.
Then
$$
k_{\varepsilon,c}(1/t)
=
t F(1/t)+\varepsilon ct.
$$
The map $t \mapsto t F(1/t)$ is strictly convex, while the second term is affine. Hence
$u\mapsto k_{\varepsilon,c}(1/t)$ is strictly convex.
Moreover,
$$
k_{\varepsilon,c}'(t)
=
\frac{
tF'(t)-F(t)-\varepsilon c
}{t^2}.
$$
Thus $k_{\varepsilon,c}$ is strictly increasing exactly when
$tF'(t)-F(t)>\varepsilon c$
for every $t>0$.
Theorem~\ref{thm:ent-admit} then shows that
$k_{\varepsilon,c}=\varepsilon g_{\varepsilon,c}$ is admissible.

Finally,
$$
\varepsilon
\left(
tg_{\varepsilon,c}(t)-g_{\varepsilon,c}(1)
\right)
=
F(t),
$$
since $F(1)=0$. A suitable choice of $c$ exists for at least one
$\varepsilon\in\{1,-1\}$ if and only if
condition~\eqref{ineq:f-div-admit} is satisfied.
\end{proof}

Proposition~\ref{prop:f-divergence-representation} identifies the precise
relationship between the present entropy admissibility framework and
Csisz\'ar $f$-divergences. Strict convexity of $F$ corresponds
exactly to the differential admissibility condition of Corollary~\ref{cor:differential-admissibility}.
The additional restriction is monotonicity of $g$, which becomes a
first-order condition on 
$t F'(t)- F(t)$.
Thus, admissible entropy generators correspond not to arbitrary
$f$-divergence generators without qualification, but to those convex
generators satisfying the  first-order condition~\eqref{ineq:f-div-admit}  in Proposition~\ref{prop:f-divergence-representation}.

\begin{prop}[Admissibility and data processing]
\label{prop:data-processing}
Let $g:(0,\infty)\to\mathbb R$ be continuous and strictly monotone, and define
$$
F_g(t)
:=
\varepsilon_g\big(tg(t)-g(1)\big),
\qquad t>0.
$$
For any measurable space $(\mathcal{X},\sigma({\mathcal{X})})$ and probability
measures $P$ and $Q$ on $\mathcal X$ satisfying $P\sim Q$, write
$$
D_{F_g}(P\Vert Q)
:=
\int_{\mathcal X}
F_g\left(\frac{\dd P}{\dd Q}\right)\dd Q,
$$
whenever the integral is finite. Then
$$
\mathcal H_g(P,(\mathcal X,Q))
=
-\log g^{-1}
\left(
g(1)+\varepsilon_gD_{F_g}(P\Vert Q)
\right).
$$

The following statements are equivalent:
\begin{enumerate}
\item The generator $g$ is admissible.
\item The function $F_g$ is strictly convex.
\item For every pair of measurable spaces
$(\mathcal{X}, \sigma(\mathcal{X}))$ and $(\mathcal{Y},\sigma(\mathcal{Y}))$, every pair of probability measures $P\sim Q$ on $\mathcal X$ for which $D_{F_g}(P\Vert Q)<\infty$, and every Markov kernel $
K:\mathcal{X}\times\sigma(\mathcal{Y})\rightarrow[0,1]$,
one has
$$
\mathcal{H}_g(PK,(\mathcal{Y},QK))
\geq
\mathcal{H}_g(P,(\mathcal{X},Q)).
$$
Here the image measures are defined by
$$
PK(B):=\int_{\mathcal X}K(x,B)\,\dd P(x),
\qquad
QK(B):=\int_{\mathcal X}K(x,B)\,\dd Q(x),
\qquad B\in \sigma(\mathcal{Y}).
$$

Moreover, let $M$ be the joint probability measure on
$\mathcal X\times\mathcal Y$,  $
M(\dd x,\dd y):=Q(\dd x)K(x,\dd y)$,
and let $\mathbf X$ and $\mathbf Y$ denote the coordinate random
variables $\mathbf{X}(x,y):=x$
and $\mathbf{Y}(x,y):=y$.
Then equality holds if and only if
$$
\frac{\dd P}{\dd Q}(\mathbf X)
=
\mathbb E_M\!\left[\frac{\dd P}{\dd Q}(\mathbf X)\mid\mathbf Y\right]
\qquad M\text{-almost surely}.
$$

\end{enumerate}
\end{prop}

\begin{remark}
The equivalence between parts $(1)$ and $(2)$ begins with a generator
$g$ that is already assumed strictly monotone. Hence
$F_g$ automatically has the entropy representation,
$F_g(t)
=
\varepsilon_g\bigl(tg(t)-g(1)\bigr)$.
By contrast, an arbitrary strictly convex $f$-divergence generator
$F$ need not admit such a representation by a monotone $g$.
Proposition~\ref{prop:f-divergence-representation} shows that the
additional necessary and sufficient condition is  condition~\eqref{ineq:f-div-admit}.
\end{remark}

\begin{proof}[Proof of Proposition~\ref{prop:data-processing}]
The equivalence of $(1)$ and $(2)$ follows from
Theorem~\ref{thm:ent-admit}. Indeed, the strict convexity or concavity of
$t\mapsto g(1/t)$ is equivalent
to the strict convexity of
$t\mapsto \varepsilon_g\,tg(t)$.

Let $R:=\frac{\dd P}{\dd Q}$. 
Then 
$$
\mathcal H_g(P,(\mathcal X,Q))
=
-\log g^{-1}
\left( \int_{\mathcal X}g(R)\,\dd P
\right).
$$
Starting from
$$
D_{F_g}(P\Vert Q)
=
\varepsilon_g
\left(
\int_{\mathcal X}g(R)\,\dd P-g(1)
\right),
$$
and rearranging leads to the displayed entropy representation.

The likelihood ratio of the image measures under $K$ satisfies
$$
\frac{\dd (PK)}{\dd (QK)}(\mathbf{Y})
=
\mathbb E_{M}[R(\mathbf{X})\mid \mathbf{Y}] 
\qquad M\text{-almost surely}.
$$
If $F_g$ is strictly convex, conditional Jensen's inequality gives
$$
D_{F_g}(PK\Vert QK)
\leq
D_{F_g}(P\Vert Q),
$$
with equality if and only if
$$
R(\mathbf{X})=\mathbb E_{M}[R(\mathbf{X})\mid \mathbf{Y}]
$$
almost surely. Since $t \mapsto
-\log g^{-1}\bigl(g(1)+\varepsilon_g t\bigr)$ is strictly decreasing, the corresponding entropy inequality and its
equality condition follow.

It remains to prove that $(3)$ implies $(2)$. Let $a,b>0$,
$\lambda\in(0,1)$, and suppose that $a\neq b$. Set
$$
m:=\lambda a+(1-\lambda)b.
$$
Choose $\delta>0$ sufficiently small that
$$
c:=\frac{1-\delta m}{1-\delta}>0,
$$
and define probability measures on $\{1,2,3\}$ by
$$
Q=
\bigl(\delta \lambda,\delta(1-\lambda),1-\delta\bigr)
$$
and
$$
P=
\bigl(\delta\lambda a,\delta(1-\lambda)b,(1-\delta)c\bigr).
$$
Then
$$
\frac{\dd P}{\dd Q}=(a,b,c).
$$
Let $K$ be the deterministic kernel induced by the map
$T:\{1,2,3\}\to\{A,B\}$ given by
$T(1)=T(2)=A$ and $T(3)=B$. Then
$$
QK=(\delta,1-\delta),
\qquad
PK=(\delta m,(1-\delta)c),
$$
and hence
$$
\frac{\dd(PK)}{\dd(QK)}=(m,c).
$$

Since $T(1)=T(2)$ but $\frac{\dd P}{\dd Q}(1)=a\neq b
=\frac{\dd P}{\dd Q}(2)$,
the likelihood ratio is not measurable with respect to
$\sigma(T)$. Hence equality in $(3)$ cannot hold. 
The strict equality
condition in $(3)$ therefore gives
$$
D_{F_g}(PK\Vert QK)<D_{F_g}(P\Vert Q).
$$
Expanding both sides yields
$$
\delta F_g(m)+(1-\delta)F_g(c)
<
\delta\bigl[\lambda F_g(a)+(1-\lambda)F_g(b)\bigr]
+(1-\delta)F_g(c),
$$
which, after cancellation and division by $\delta$, gives
$$
F_g\bigl(\lambda a+(1-\lambda)b\bigr)
<
\lambda F_g(a)+(1-\lambda)F_g(b).
$$
Since $a,b>0$ and $\lambda\in(0,1)$ were arbitrary, $F_g$ is strictly
convex.
\end{proof}

\begin{prop}[Minimum-divergence estimating equation]
\label{prop:min-div-estimating}
Let $g\in C^1(0,\infty)$ be an admissible generator, and let $F_g$ be defined as in Proposition~\ref{prop:data-processing}. For $\Theta \subseteq \mathbb{R}^d$, let $\{Q_\theta:\theta \in \Theta\}$ be a differentiable family of
probability measures having strictly positive densities $q_\theta$ with respect to a measure $\lambda$, and let $P$ be a probability measure with a strictly positive density $p:= \dd P/\dd \lambda$. Write
$$
R_\theta:=\frac{p}{q_\theta},
\qquad
s_\theta:=\nabla_\theta\log q_\theta.
$$
Suppose that differentiation under the integral is justified. If
$\theta^\ast$ is an interior minimizer of $\theta\mapsto D_{F_g}(P\Vert Q_\theta)$,
then
$$
\int_{\mathcal X}
w_g(R_{\theta^\ast})s_{\theta^\ast}\,\dd P
=
0,
$$
where $w_g(t):=t|g'(t)|$.
\end{prop}

\begin{proof}
Differentiation gives
$$
\nabla_\theta D_{F_g}(P\Vert Q_\theta)
=
\int_{\mathcal X}
q_\theta
\left[
F_g(R_\theta)
-
R_\theta F_g'(R_\theta)
\right]
s_\theta\,\dd\lambda.
$$
Moreover,
$$
R_\theta F_g'(R_\theta)-F_g(R_\theta)
=
\varepsilon_g
\left(
R_\theta^2g'(R_\theta)+g(1)
\right).
$$
Noting that $\int_{\mathcal X}q_\theta s_\theta\,\dd\lambda=0$, 
$\varepsilon_g g'(t) = |g'(t)|$, and $q_\theta R_\theta^2 = R_\theta p$,
the first-order condition reduces to
$$
\int_{\mathcal X}
R_\theta|g'(R_\theta)|s_\theta\,\dd P
= 0,
$$
which completes the proof.
\end{proof}

Proposition~\ref{prop:min-div-estimating} specializes the classical minimum-$f$-divergence first-order condition to the entropy-generated class $F_g$, identifying
$w_g(t)=t|g'(t)|$ as the density-ratio weight induced by the admissible generator $g$.
For the arctangent generator constructed below, $g(t)=\arctan t$, one has
$$
w_g(t)=\frac{t}{1+t^2},
$$
and the corresponding first-order condition becomes
$$
\int_{\mathcal X}
\frac{R_{\theta^\ast}}{1+R_{\theta^\ast}^2}
s_{\theta^\ast}\,dP = 0.
$$
Here $w_g(t)\leq 1/2$ and $w_g(t)\to0$ as either $t\downarrow0$ or $t\uparrow\infty$, so extreme density-ratio regimes are downweighted.
Robustness and asymptotic properties of sample-based estimators require additional model-specific analysis.

We conclude this section with three constructions illustrating how the framework produces both classical and non-power entropy families.
\\

\paragraph{\bf Sharma-Mittal entropy.}
The Sharma-Mittal family \cite{SharMitt75,SharMitt77,Masi05} illustrates the separation between
the generalized-mean generator, $g$, and the outer entropy scale, $h$. For
$q>0$, $q\neq 1$, and $r>0$, $r\neq 1$, take
$$
g_q(t)=t^{q-1},
\qquad
h_r(u)=\frac{1-u^{r-1}}{r-1}.
$$
Since $h_r'(u)=-u^{r-2}<0$,
Definition~\ref{def:gh-ent} gives
\begin{equation}
H_{\nu}^{q,r}(\mu,\mathcal{X})
=
\frac{1}{r-1}
\left[
1-
\left(
\frac{1}{\mu(\mathcal{X})}
\int_{\mathcal{X}}
\left(
\frac{1}{\mu(\mathcal{X})}
\frac{d\mu}{d\nu}
\right)^{q-1}
\,d\mu
\right)^{\frac{r-1}{q-1}}
\right]
=
\frac{1}{r-1}
\left[
1-
\left(
\frac{1}{\mu(\mathcal{X})^q}
\int_{\mathcal{X}}
\left(
\frac{d\mu}{d\nu}
\right)^q
\,d\nu
\right)^{\frac{r-1}{q-1}}
\right].
\end{equation}

The limit $r\rightarrow 1$ gives R\'enyi entropy of order $q$, while $r=q$ gives Tsallis entropy. These functionals therefore retain the same admissible power-mean generator while varying the outer entropy
scale.
\\

\paragraph{\bf Arctangent entropy.}
Let $g(t)=\arctan t$ and $h=-\log$.
Then $g$ is strictly increasing and
$$
2g'(t)+t g''(t)
=
\frac{2}{(1+t^2)^2}>0,
$$
so $g$ is admissible. Since $g^{-1}(y)=\tan y$,
the resulting entropy is
$$
H_{\arctan}(\mu,(\mathcal{X},\nu)) =
-\log
\left[
\tan
\left(
\frac{1}{\mu(\mathcal{X})}
\int_{\mathcal{X}}
\arctan
\left(
\frac{1}{\mu(\mathcal{X})}
\frac{d\mu}{d\nu}
\right)
\,d\mu
\right)
\right].
$$

\paragraph{\bf Integral transform entropies.}
A broad class of examples is obtained from integral transforms. For these we also take $h=-\log$.

\begin{cor}[Positive-kernel construction]
\label{prop:kernel-construction}
Let $(S,\mathcal{S},\Lambda)$ be a finite positive measure space, and let
$$
K:S\times(0,\infty)\longrightarrow\mathbb{R}
$$
be measurable in $s$ and twice continuously differentiable in $t$. Assume
that the first two derivatives in $t$ may be passed under the integral sign
and that, for every $t>0$,
$$
\int_S \partial_tK(s,t)\,\dd\Lambda(s)>0
$$
and
$$
\int_S
\left(
2\partial_tK(s,t)+t\partial_{tt}K(s,t)
\right)
\dd\Lambda(s)>0.
$$
Then
$$
g_\Lambda(t):=\int_S K(s,t)\,\dd\Lambda(s)
$$
is an increasing admissible entropy generator. If both displayed
inequalities are reversed, then $g_\Lambda$ is a decreasing admissible
entropy generator.
\end{cor}

\begin{proof}
Differentiation under the integral gives
$$
g_\Lambda'(t)
=
\int_S\partial_tK(s,t)\,\dd\Lambda(s)
$$
and
$$
2g_\Lambda'(t)+t g_\Lambda''(t)
=
\int_S
\left(
2\partial_tK(s,t)+t\partial_{tt}K(s,t)
\right)
\dd\Lambda(s).
$$
The result follows from
Corollary~\ref{cor:differential-admissibility}.
\end{proof}

An explicit example is given by the Laplace transform, with $\Lambda((0,\infty))>0$,
$$
g_{\Lambda}(t)
:=
\int_{[0,\infty)}
\exp\left(-\frac{s}{t}\right)
\,d\Lambda(s),
\qquad t>0.
$$
It is straightforward to check that the kernel
$K(s,t)=\exp(-s/t)$ satisfies the conditions of
Corollary~\ref{prop:kernel-construction}; hence $g_\Lambda$ is admissible. The corresponding entropy is given by
$$
H_{\Lambda}(\mu,(\mathcal{X},\nu))
=
-\log
g_{\Lambda}^{-1}
\left[
\frac{1}{\mu(\mathcal{X})}
\int_Y
\left(
\int_{[0,\infty)}
\exp\left(
-\frac{s\mu(\mathcal{X})}
{d\mu/d\nu}
\right)
\,d\Lambda(s)
\right)
\,d\mu
\right],
$$
where  $Y :=
\left\{
x\in\mathcal{X}:
\frac{d\mu}{d\nu}(x)>0
\right\}$. 
Thus every Laplace transform of the above form gives an admissible entropy functional.

The Laplace transform averages exponential attenuation over the scales encoded by $\Lambda$. Correspondingly, $H_{\Lambda}$ applies these
attenuation scales to the reciprocal normalized Radon--Nikodym density,
averages the resulting values with respect to $\mu$, and then returns the
result to the original scale through $g_{\Lambda}^{-1}$. This gives an
infinite-dimensional family of admissible entropies parameterized by
positive measures $\Lambda$.

Laplace transforms provide one natural class;
Mellin-type and other transform constructions may provide further examples, although admissibility must be verified separately in each case.

As the simplest instance, take
$\Lambda=\delta_{\lambda}$ and $\lambda>0$. Then
$$
g_{\lambda}(t)
=
\exp\left(-\frac{\lambda}{t}\right),
\qquad
g_{\lambda}^{-1}(y)
=
-\frac{\lambda}{\log y},
$$
and the corresponding entropy is
$$
H_{\lambda}(\mu,(\mathcal{X},\nu))
=
\log
\left[
-\frac{1}{\lambda}
\log
\left(
\frac{1}{\mu(\mathcal{X})}
\int_Y
\exp\left(
-\frac{\lambda\mu(\mathcal{X})}
{d\mu/d\nu}
\right)
\,d\mu
\right)
\right].
$$
The associated density-ratio weight is
$$
t g_{\lambda}'(t)
=
\frac{\lambda}{t}
\exp\left(-\frac{\lambda}{t}\right),
$$
which is bounded and redescending.

This explicit case illustrates the construction without introducing a
separate family beyond the positive-kernel proposition.
Table~\ref{tab:additional-generators} in Sec.~\ref{sec:proofs} includes other admissible generators and their corresponding entropy functionals.

\section{Entropies and Generalized Means}

The Shannon and R\'enyi entropies introduced in Sec.~\ref{sec:main} have a surprisingly direct relation to the power means \cite{Kol30, Nagumo1930, Bullen2003,deCar16}, given by Definition~\ref{def:Lp-mean}. We review several relevant properties of such means.

\label{sec:ent-mean-lem}
\begin{lem}
\label{prop:ent-max}
For every pair of nonzero finite measures $\mu\ll\nu$ on $\mathcal{X}$ and every $p\geq1$ for which the displayed entropies are well defined, Shannon entropy dominates R\'enyi entropy: 
\begin{equation}
\label{ineq:ent-max}
 H^p_\nu(\mu,\mathcal{X})  \leq S_\nu(\mu,{\mathcal{X}}) \leq  S_\nu(c \nu,{\mathcal{X}}), 
\end{equation}
where $c>0$ and $S_\nu(c\nu,{\mathcal{X}}) = \log \nu(\mathcal{X})$. Thus the maximum possible entropy is  $\log\nu(X)$, and it is attained exactly by measures proportional to the reference measure, that is, for $\mu = c\nu$, where $c>0$.
\end{lem}

\begin{proof}
For $p=1$, the first inequality is an equality. Let $p>1$.
Then Jensen's inequality gives,
\begin{align*}
H^p_\nu(\mu,\mathcal{X}) 
&=  \log \mu(\mathcal{X}) - \dfrac{1}{p-1 }  \log \left( 
 \frac{1}{\mu({\mathcal{X}})} \int_{{\mathcal{X}}} \left(\dfrac{\dd \mu}{\dd \nu} \right)^{p-1}
 \dd \mu \right)   \\
 &\leq  \log \mu(\mathcal{X}) - \dfrac{1}{p-1 }   \left( 
 \frac{1}{\mu({\mathcal{X}})} \int_{{\mathcal{X}}} \log\left(\dfrac{\dd \mu}{\dd \nu} \right)^{p-1}
 \dd \mu \right)   \\
&=  \log \mu(\mathcal{X}) -  
 \frac{1}{\mu({\mathcal{X}})} \int_{{\mathcal{X}}} \log\left(\dfrac{\dd \mu}{\dd \nu} \right)
 \dd \mu   \\
&= S_\nu(\mu,{\mathcal{X}}) \, ,
\end{align*}
with equality if and only if $\dd \mu / \dd \nu$ is constant $\mu$-almost everywhere.

For the second inequality we let
$Y:=\left\{x \in \mathcal{X}: \dfrac{\dd \mu}{\dd \nu}(x)>0  \right\}$.
Then, starting from (\ref{eq:ent-def}) and denoting $\mu' = \dd\mu/\dd \nu$, 
\begin{align*}
S_\nu(\mu,{\mathcal{X}}) 
 &= \log \mu(\mathcal{X}) 
 + \frac{1}{\mu({\mathcal{X}})} \int_{Y} \mu' 
 \log \left(\dfrac{1}{\mu'} \right)  \dd \nu  \\
&\leq \log \mu(\mathcal{X}) 
 + \log \left[    \frac{1}{\mu({\mathcal{X}})} \int_{Y}    \mu' 
  \left(\dfrac{1}{ \mu'} \right)  \dd \nu   \right] \\
&= \log \mu(\mathcal{X}) 
 + \log   \frac{\nu( Y)}{\mu({\mathcal{X}})} \\
&= \log \nu(Y) \\
&\leq \log \nu({\mathcal{X}})  \\
&=S_\nu(c\nu,{\mathcal{X}}) \, ,
\end{align*}
where equality holds if and only if $\nu(Y)=\nu(\mathcal{X})$ and $d\mu/d\nu$ constant
$\mu$-almost everywhere, or equivalently, $\mu=c\nu$ for some $c>0$.
\end{proof}

\begin{lem}
\label{lem:M0}
Let $f>0$ and $f^{\epsilon}, f^{-\epsilon} \in L^1(\mu)$ for some $\epsilon>0$. Then the $L^p$ mean of order $0$ is the geometric mean, given by 
\begin{equation}
M^0_\mu (f) :=\lim_{p \rightarrow 0} M^p_\mu (f) 
= \exp\left( \dfrac{1}{\mu(\mathcal{X})}\int_\mathcal{X}  \log f \; \dd \mu \right).
\end{equation}
\end{lem}
\begin{proof}
Taking logarithms we have
\begin{align*}
\log M^0_\mu (f) &=\lim_{p \rightarrow 0} \log M^p_\mu (f) \\
&= \lim_{p \rightarrow 0} \frac{1}{p} \log \left(   \dfrac{1}{\mu(\mathcal{X})}  \int_\mathcal{X}  f^p \, \dd \mu \right) \\
&= \lim_{p \rightarrow 0} \left( \dfrac{ \dfrac{1}{\mu(\mathcal{X})}  \int_\mathcal{X}  f^p \log f \, \dd \mu}{  \dfrac{1}{\mu(\mathcal{X})}  \int_\mathcal{X}  f^p \, \dd \mu }  \right)  \\
&=\dfrac{1}{\mu(\mathcal{X})}  \int_\mathcal{X}   \log f \, \dd \mu  \, ,
\end{align*}
where we used L'Hôpital's rule in the third line.
\end{proof}

\begin{remark}
\label{rem:ren-ent-M}
We may write the R\'enyi entropy as 
$$
H^p_\nu(\mu,\mathcal{X}) =  - \log M_\mu^{p-1} \left(\dfrac{1}{\mu(\mathcal{X}) } \dfrac{\dd \mu}{\dd \nu}, \mathcal{X} \right) 
=  \log \mu(\mathcal{X}) - \log M_\mu^{p-1} \left( \dfrac{\dd \mu}{\dd \nu} \right) .
 $$
 Together with Lemma~\ref{lem:M0}, this implies that 
$H^1_\nu(\mu, \mathcal{X}) = S_\nu(\mu, \mathcal{X})$.
\end{remark}

\begin{lem}
\label{lem:M-monotone}
For $r<s$, let $p\mapsto \int f^p\,d\mu$ be differentiable on $(r,s)$ and $\int f^p|\log f|\,d\mu<\infty$ for $p \in (r-\varepsilon,s+\varepsilon)$ and some $\varepsilon>0$. If  $r\leq 0$, assume in addition that $f>0$ $\mu$-almost everywhere. Then
\begin{equation}
\label{eq:ent-r-s-ineq}
M_\mu^r(f)\leq M_\mu^s(f),
\end{equation}
with equality if and only if $f$ is constant
$\mu$-almost everywhere.
\end{lem}
\begin{proof}
Let $\bar\mu:=\frac{\mu}{\mu(\mathcal X)}$.
Then,
$$
\begin{aligned}
\frac{\partial}{\partial p}\log M^p_\mu(f)
&= \frac{\partial}{\partial p} \left(  \frac{1}{p}\log \int_{\mathcal{X}} f^p \dd \bar{\mu}  \right)\\
&=  -\frac{1}{p^2} \log \int_{\mathcal{X}} f^p \;\dd \bar{\mu}  + \frac{1}{p} \dfrac{ \int_{\mathcal{X}} f^p \log f\;\dd \bar{\mu} }{  \int_{\mathcal{X}} f^p \;\dd \bar{\mu}  }  \\
&=\frac{1}{p^2} \left( - \log \int_{\mathcal{X}} f^p \;\dd \bar{\mu}  + p \cdot \dfrac{ \int_{\mathcal{X}} f^p \log f\;\dd \bar{\mu} }{  \int_{\mathcal{X}} f^p \;\dd \bar{\mu}  } \right)
\end{aligned}
$$

Defining the probability measure $\nu_p$ by
$$
\dd\nu_p\
:=
\frac{f^p}{\int_{\mathcal{X}} f^p \;\dd \bar{\mu} } \dd\bar\mu
$$
we continue the above computation by writing
$$
\begin{aligned}
\frac{\partial}{\partial p}\log M^p_\mu(f)
&= \frac{1}{p^2} \left( - \log \int_{\mathcal{X}} f^p \;\dd \bar{\mu}  + \int_{\mathcal{X}}  \log f^p\;\dd\nu_p  \right) \\
&= \frac{1}{p^2}  \int_{\mathcal{X}} \log \left(\frac{\dd\nu_p}{\dd\bar\mu}\right) \dd\nu_p
\end{aligned}
$$
which is nonnegative by Jensen's inequality. 
Since the logarithm is strictly increasing, it follows that $p\mapsto M^p_\mu(f)$ is nondecreasing on every interval on which the above calculation is valid.
Moreover, equality holds if and only if $\nu_p=\bar\mu$.
Since $p\neq0$, this holds if and only if
$f$ is constant $\mu$-almost everywhere.
Thus, for $r<s$ lying on the same side of zero, $M^r_\mu(f)\leq M^s_\mu(f)$, with equality if and only if $f$ is constant $\mu$-almost everywhere.
The cases in which $r\leq0\leq s$ follow by continuity at $p=0$ as given by Lemma~\ref{lem:M0}, where
$$
M^0_\mu(f)
=
\exp\left(
\int_{\mathcal X}\log f\,\dd\bar\mu
\right).
$$
Therefore, whenever the relevant means are defined, $M^r_\mu(f)\leq M^s_\mu(f)$ for all $r<s$,
with equality if and only if $f$ is constant $\mu$-almost everywhere.
\end{proof}

Lemma~\ref{lem:M-monotone} then implies the following result for Rényi entropies, which generalizes Lemma~\ref{prop:ent-max}.
\begin{cor}
For all $0 \leq r < s$, we have $H^s_\nu(\mu,\mathcal{X}) \leq H^r_\nu(\mu,\mathcal{X})$.
\end{cor}

\begin{lem}
\label{lem:M}
For two measures $\mu, \nu$, with  $\mu \ll \nu$, on a space $\mathcal{X}$, and for $p > 1$, we have
\begin{align}
\mu(\mathcal{X}) &\leq M^p_\nu \left( \dfrac{\dd \mu}{\dd \nu}  \right) \nu(\mathcal{X}) \\
\label{eq:exp-renyi-ineq}
\mu(\mathcal{X}) &\leq  M^{p-1}_\mu \left( \dfrac{\dd \mu}{\dd \nu}  \right) \nu(\mathcal{X}) \, .
\end{align}
\end{lem}
\begin{proof}
For the first inequality, we have by Hölder's inequality for $p,q \geq 1$ satisfying $\dfrac{1}{p} + \dfrac{1}{q} = 1$,
\begin{align*}
\mu(\mathcal{X}) &= \int_\mathcal{X} \dfrac{\dd \mu}{\dd \nu} \dd \nu \\
&\leq   \left(  \int_\mathcal{X}  \left(  \dfrac{\dd \mu}{\dd \nu} \right)^p \dd \nu \right)^{\frac{1}{p}} \left( \int_\mathcal{X} 1^q \dd \nu  \right)^{\frac{1}{q}}      \\
&= \left(   \dfrac{1}{\nu(\mathcal{X})} \int_\mathcal{X}   \left(  \dfrac{\dd \mu}{\dd \nu} \right)^p  \, \dd \nu \right)^{\frac{1}{p}}      \nu(\mathcal{X})^{\frac{1}{p}}     \nu(\mathcal{X}) ^{\frac{1}{q}}  \\
&= M^p_\nu \left( \dfrac{\dd \mu}{\dd \nu}  \right) \nu(\mathcal{X}) \, .
\end{align*}
For the second inequality, we start from the second line in the previous computation, which gives, \begin{align*}
\mu(\mathcal{X}) =\mu(\mathcal{X})^{\frac{1}{p}} \mu(\mathcal{X})^{\frac{1}{q}}
&\leq   \left(  \int_\mathcal{X}  \left(  \dfrac{\dd \mu}{\dd \nu} \right)^p \dd \nu \right)^{\frac{1}{p}} \nu(\mathcal{X})^{\frac{1}{q}}      \\
&= \left(  \dfrac{1}{\mu(\mathcal{X})} \int_\mathcal{X}  \left(  \dfrac{\dd \mu}{\dd \nu} \right)^p \dd \nu \right)^{\frac{1}{p}} \mu(\mathcal{X})^{\frac{1}{p}} \, \nu(\mathcal{X})^{\frac{1}{q}}   \, ,
\end{align*}
whereupon, canceling $\mu(\mathcal{X})^{\frac{1}{p}}$ on both sides of the inequality and exponentiating by $q$, we obtain,
\begin{align*}
\mu(\mathcal{X}) 
&\leq   \left(  \dfrac{1}{\mu(\mathcal{X})} \int_\mathcal{X}  \left(  \dfrac{\dd \mu}{\dd \nu} \right)^p \dd \nu \right)^{\frac{q}{p}} \, \nu(\mathcal{X})    \\
&=    \left(  \dfrac{1}{\mu(\mathcal{X})} \int_\mathcal{X}  \left(  \dfrac{\dd \mu}{\dd \nu} \right)^{p-1}  \dfrac{\dd \mu}{\dd \nu} \, \dd \nu \right)^{\frac{q}{p}} \, \nu(\mathcal{X})    \\
&=   \left[ \left(  \dfrac{1}{\mu(\mathcal{X})} \int_\mathcal{X}  \left(  \dfrac{\dd \mu}{\dd \nu} \right)^{p-1}     \dd \mu \right)^{\frac{1}{p-1}} \right]^{\frac{q(p-1)}{p}} \, \nu(\mathcal{X})    \\
&= M^{p-1}_\mu \left( \dfrac{\dd \mu}{\dd \nu}  \right) \nu(\mathcal{X}) \, ,
\end{align*}
where we used the fact that $\frac{q(p-1)}{p} = 1$.
\end{proof}

The following two lemmas establish two different multiplicativity, or factorization, properties for $L^p$ means.
\begin{lem}
\label{lem:mean-factorization}
For any measurable positive functions $f,f'$ on a finite measure space $(\mathcal{X},\mu)$, the geometric mean is the unique power mean satisfying 
\begin{equation}
\label{eq:M-internal-add}
M^p_\mu(ff',\mathcal{X})=M^p_\mu(f,\mathcal{X})M^p_\mu(f',\mathcal{X}).
\end{equation}
\end{lem}
\begin{proof}
Consider first the case $p=0$. Then by Lemma~\ref{lem:M0}, we have 
$$
M^0_\mu(ff')
=
\exp\left(\frac{1}{\mu(\mathcal{X})}\int_{\mathcal{X}}\log(ff')\,d\mu\right)
=
M^0_\mu(f)M^0_\mu(f'),
$$
confirming that the geometric mean satisfies \eqref{eq:M-internal-add}.
Conversely, let $p\neq0$ and consider the two-point space $X=\{1,2\}$ with equal weights
and let $f=(a,b)$, $f'=(b,a)$ for $a,b>0$. Then $ff'=(ab,ab)$, and therefore
$$
M^p_\mu(ff')=
\left(\frac{(ab)^p+(ab)^p}{2}\right)^{1/p}=ab.
$$
On the other hand,
$$
M^p_\mu(f)M^p_\mu(f')
=
\left(\frac{a^p+b^p}{2}\right)^{1/p}
\left(\frac{b^p+a^p}{2}\right)^{1/p}
=
\left(\frac{a^p+b^p}{2}\right)^{2/p}.
$$
For equality to hold for all $a,b>0$, we would need
$$
\frac{a^p+b^p}{2}=(ab)^{p/2}
$$
for all $a,b>0$, which is false whenever $a \neq b$, by the strict
AM--GM inequality applied to $a^p$ and $b^p$. Hence no $p \neq 0$
power mean satisfies the internal factorization property universally.
\end{proof}

Although Lemma \ref{lem:mean-factorization} shows that, other than the geometric mean, the other $L^p$ means do not generally factor in terms of their input functions, they do satisfy a weaker factorization property for product spaces.
\begin{lem}
\label{lem:prod-space-mean}
    Given measurable functions $f$ on the measure space $(\mathcal{X}, \mu)$ and $f'$ on $(\mathcal{Y}, \nu)$, we have
    \begin{equation}
        \label{eq:mean-factor-product}
        M^p_{\mu \otimes \nu} \big(f \cdot f', \mathcal{X} \times \mathcal{Y} \big) = M^p_{\mu } \big(f , \mathcal{X}  \big) \cdot M^p_{\nu} \big( f',  \mathcal{Y} \big) 
    \end{equation}
    for any $p \in \mathbb{R}$, where we take $f,f'>0$ if $p\leq 0$.
\end{lem}
\begin{proof}
For $p \neq 0$, we obtain from direct computation,
\begin{align*}
    M^p_{\mu \otimes \nu} \big(f \cdot f', \mathcal{X} \times \mathcal{Y} \big) 
    &=\left( \dfrac{1}{\mu(\mathcal{X}) \nu(\mathcal{Y}) }  \int_\mathcal{X} \int_\mathcal{Y} f(x)^p ~f'(y)^p ~ \dd \mu(x) ~\dd \nu(y)   \right)^{ \frac{1}{p}} \\
    &=  \left( \dfrac{1}{\mu(\mathcal{X}) }  \int_\mathcal{X}  f(x)^p  ~ \dd \mu(x)   ~\cdot  \dfrac{1}{ \nu(\mathcal{Y}) }  \int_\mathcal{Y}  ~f'(y)^p ~  ~\dd \nu(y)   \right)^{ \frac{1}{p}} \\
     &= M^p_{\mu } \big(f , \mathcal{X}  \big) \cdot M^p_{\nu} \big( f',  \mathcal{Y} \big) .
\end{align*}
For $p = 0$, we use Lemma~\ref{lem:M0}, which gives
\begin{align*}
    \log M^0_{\mu \otimes \nu} \big(f \cdot f', \mathcal{X} \times \mathcal{Y} \big)
    &= \dfrac{1}{\mu(\mathcal{X}) \nu(\mathcal{Y}) }  \int_\mathcal{X} \int_\mathcal{Y} \log \big(f(x) \;f'(y) \big) ~ \dd \mu(x) ~\dd \nu(y)   \\
    &= \dfrac{1}{\mu(\mathcal{X}) \nu(\mathcal{Y}) }  \int_\mathcal{X} \int_\mathcal{Y} \big( \log f(x) +  \log f'(y) \big) ~ \dd \mu(x) ~\dd \nu(y)    \\
     &=   \dfrac{1}{\mu(\mathcal{X}) }  \int_\mathcal{X}  \log f  ~ \dd \mu     +   \dfrac{1}{ \nu(\mathcal{Y}) }  \int_\mathcal{Y}  \log f' ~  \dd \nu   \\
     &= \log M^0_{\mu } \big(f , \mathcal{X}  \big) + \log M^0_{\nu} \big( f',  \mathcal{Y} \big) ,
\end{align*}
which completes the proof.
\end{proof}

\begin{remark}
\label{rem:multip-mean-char}
   Notice that, for a constant $c>0$, any generalized mean satisfying the factorization of Lemma~\ref{lem:prod-space-mean} must also satisfy
   $$
\mathcal{M}_g(f \otimes c) = \mathcal{M}_g(f ) \mathcal{M}_g( c) = c  \mathcal{M}_g(f ).
   $$
   Thus, product multiplicativity implies homogeneity of the generalized mean. Characterization of homogeneous generalized means then gives that 
   $g(t) = a \log t +b$ or $g(t) = a t^p +b$, corresponding to the geometric and power means, respectively.  
   Thus, the $L^p$ means are uniquely characterized as the means satisfying the multiplicative property of Lemma \ref{lem:prod-space-mean} \cite{FerPalPer08, AubNech11, Leinster12}. 
\end{remark}

\begin{remark}
\label{remark:ent-max}
Referring to formula (\ref{eq:exp-renyi-ineq}) of Lemma \ref{lem:M}, we note taking the logarithm of both sides and rearranging gives the upper bound of the Rényi entropy:
\begin{align*}
H_\nu^p(\mu, \mathcal{X}) = \log \mu(\mathcal{X}) -\log  M^{p-1}_\mu \left( \dfrac{\dd \mu}{\dd \nu}  \right) 
\leq \log \nu(\mathcal{X}) \, .
\end{align*}
This was derived in Lemma \ref{prop:ent-max} using Jensen's inequality.
\end{remark}

\begin{lem}
\label{lem:4int-4}
For generalized $g$-entropies, Postulate $4_{\rm int}$ implies Postulate $4$.
\end{lem}

\begin{proof}
Let $\mu_i\ll\nu_i$ be finite measures on $\mathcal{X}_i$ for $i=1,2$,
and define the functions, for $i=1,2$,
$$
f_i = \frac1{\mu_i(\mathcal{X}_i)}\frac{d\mu_i}{d\nu_i} \;.
$$

When working on the product space $\mathcal{X}_1\times\mathcal{X}_2$, we use a slight abuse of notation for clarity, writing $f_i$ for the pullback $f_i\circ\pi_i$, where
$
\pi_i:\mathcal X_1\times\mathcal X_2\to\mathcal X_i
$
is the coordinate projection.

Then, by Postulate $4_{\rm int}$, 
 \begin{equation}
 \label{eq:Hg-f1f2}
    \widetilde{\mathcal{H}}_g(f_1f_2,(\mathcal{X}_1 \times \mathcal{X}_2,\mu_1 \otimes \mu_2))
    =
    \widetilde{\mathcal{H}}_g(f_1,(\mathcal{X}_1 \times \mathcal{X}_2,\mu_1 \otimes \mu_2))
    +
    \widetilde{\mathcal{H}}_g(f_2,(\mathcal{X}_1 \times \mathcal{X}_2,\mu_1 \otimes \mu_2)).
\end{equation}

Notice that, for the product measure, we have
$$
\frac{1}{(\mu_1\otimes\mu_2)(\mathcal{X}_1 \times \mathcal{X}_2)}
\frac{d(\mu_1\otimes\mu_2)}{d(\nu_1\otimes\nu_2)}
=
\frac{1}{\mu_1(\mathcal{X}_1)\mu_2(\mathcal{X}_2)}
\frac{d\mu_1}{d\nu_1}
\frac{d\mu_2}{d\nu_2}
=
f_1 f_2.
$$

Therefore, 
\begin{equation}
\label{eq:tildeH-H-12}
\widetilde{\mathcal{H}}_g(f_1f_2,(\mathcal{X}_1 \times \mathcal{X}_2,\mu_1 \otimes \mu_2))
= 
\mathcal{H}_g(\mu_1 \otimes \mu_2,(\mathcal{X}_1 \times \mathcal{X}_2,\nu_1 \otimes \nu_2)).
\end{equation}

Moreover,
$$
\begin{aligned}
\mathcal{M}_g(f_1,(\mathcal{X}_1\times \mathcal{X}_2,\mu_1\otimes\mu_2))
&=
g^{-1}
\left(
\frac{1}{\mu_1(\mathcal{X}_1)\mu_2(\mathcal{X}_2)}
\int_{\mathcal{X}_2} \int_{\mathcal{X}_1}
g(f_1(x_1)) \; d\mu_1(x_1)\,d\mu_2(x_2)
\right) \\
&=
g^{-1}
\left(
\frac{1}{\mu_1(\mathcal{X}_1)}
\int_{\mathcal{X}_1} 
g(f_1(x_1)) \; d\mu_1(x_1) \right)\\
&=
\mathcal{M}_g(f_1,(\mathcal{X}_1,\mu_1)), 
\end{aligned}
$$
which implies that
$$
\begin{aligned}
\widetilde{\mathcal{H}}_g(f_1,(\mathcal{X}_1\times \mathcal{X}_2,\mu_1\otimes\mu_2))
&=
-\log \mathcal{M}_g(f_1,(\mathcal{X}_1\times \mathcal{X}_2,\mu_1\otimes\mu_2)) \\
&=
-\log \mathcal{M}_g(f_1,(\mathcal{X}_1,\mu_1)) \\
&=
\mathcal{H}_g(\mu_1,(\mathcal{X}_1,\nu_1)), 
\end{aligned}
$$
and similarly for $f_2$. 
Therefore, putting this together with Eq.~\ref{eq:tildeH-H-12}, Eq.~\ref{eq:Hg-f1f2} becomes
$$
\mathcal{H}_g(\mu_1 \otimes \mu_2,(\mathcal{X}_1 \times \mathcal{X}_2,\nu_1 \otimes \nu_2))
= \mathcal{H}_g(\mu_1,(\mathcal{X}_1,\nu_1)) 
+
\mathcal{H}_g(\mu_2,(\mathcal{X}_2,\nu_2)),
$$
which is Postulate 4.
\end{proof}

\vspace{1em}

\section{Proofs of the Main Theorems and Related Results}
\label{sec:proofs}

\begin{proof}[Proof of Theorem~\ref{thm:ent-admit}]

We start with (a). We first note that every positive measurable $f$ on $(\mathcal{X},\mu)$ can be realized
as such a normalized Radon--Nikodym density by defining
$$
d\nu_f=\frac{1}{\mu(X)f}\,d\mu,
$$
on the natural domain where this defines a finite measure, so we may equivalently prove condition \eqref{eq:M-condition} for the corresponding density. 
Moreover, substituting $\mu = \nu$ in Eq.~\eqref{eq:gh-entropy} gives
\begin{equation}
\label{eq:max-hg-ent}
\mathcal H_{g,h}(\nu,(\mathcal X,\nu))
=
h\left(\frac{1}{\nu(\mathcal X)}\right).
\end{equation}
(For $h=-\log$ this is simply $\mathcal H_g(\nu,(\mathcal X,\nu)) = \log \nu(\mathcal X)$.)
More generally, the same value is obtained for any $\mu = c\,\nu$, $c>0$. 
Equation~\eqref{eq:max-hg-ent} verifies the reference-mass clause in Postulate~A and shows that every measure $\mu=c\nu$ attains equality in its first clause. The converse equality statement will follow below from the strict equality condition in \eqref{eq:M-condition}, together with $\nu(Y)=\nu(\mathcal X)$.

Since $\mu(\mathcal{X}\setminus Y)=0$, we may regard $\mu$ as a measure on
$Y$, and $\mu\ll\nu|_Y$. Applying Postulate A on the reference space
$(Y,\nu|_Y)$ gives
$$
\mathcal{H}_{g,h}(\mu,(Y,\nu|_Y))
\leq
\mathcal{H}_{g,h}(\nu|_Y,(Y,\nu|_Y)).
$$
Equivalently,
\begin{equation}
\label{eq:h-monotone}
h\left(
\mathcal{M}_g \left(
\frac{1}{\mu(\mathcal X)}
\frac{d\mu}{d\nu|_Y},
(Y,\mu)
\right)
\right)
\leq
h\left(\frac{1}{\nu(Y)}\right).
\end{equation}

Since $\mu(Y)=\mu(X)$ and $\frac{d\mu}{d(\nu|_Y)} =\frac{d\mu}{d\nu}$
on $Y$,  
$$
\mathcal{M}_g \left(
\frac{1}{\mu(\mathcal X)}
\frac{d\mu}{d\nu|_Y},
(Y,\mu)
\right) = \mathcal{M}_g \left(
\frac{1}{\mu(\mathcal{X})}
\frac{d\mu}{d\nu},
(\mathcal{X},\mu)
\right) .
$$
Since $h$ is strictly decreasing,  we then obtain
\begin{equation}
\label{ineq:M-condition}
\mathcal{M}_g \left(
\frac{1}{\mu(\mathcal{X})}
\frac{d\mu}{d\nu},
(\mathcal{X},\mu)
\right)
\geq 
\frac{1}{\nu(Y)}
\geq
\frac{1}{\nu(\mathcal{X})},
\end{equation}
where the last inequality follows because $\nu(Y) \leq \nu(\mathcal{X})$.

Now, the harmonic mean
\begin{align}
\label{eq:harmonic-mean}
    M_\mu^{-1}\left(\dfrac{1}{\mu(\mathcal{X}) }\dfrac{\dd \mu}{\dd \nu}, \mathcal{X} \right) 
    \nonumber
    &=M_\mu^{-1}\left(\dfrac{1}{\mu(\mathcal{X}) }\dfrac{\dd \mu}{\dd \nu}, Y \right)  \\
    \nonumber
    &= \left( \dfrac{1}{\mu(\mathcal{X})}\int_{Y}  \left( \dfrac{1}{\mu(\mathcal{X}) }\dfrac{\dd \mu}{\dd \nu} \right)^{-1} \dd \mu \right)^{-1} \\
    &= \left(\int_{Y}  \left( \dfrac{\dd \mu}{\dd \nu} \right)^{-1} \dfrac{\dd \mu}{\dd \nu} \, \dd \nu \right)^{-1} 
    = \dfrac{1}{\nu(Y)}.
\end{align}
Combining \eqref{ineq:M-condition} and \eqref{eq:harmonic-mean} gives the condition \eqref{eq:M-condition} for the generalized entropy. Notice that if we were to apply Postulate A directly to $\mu \ll \nu$, then Eq.~\eqref{ineq:M-condition} would have instead been the weaker result, $\mathcal{M}_g \left(
\frac{1}{\mu(\mathcal{X})}
\frac{d\mu}{d\nu},
(\mathcal{X},\mu)
\right)
\geq 
\frac{1}{\nu(\mathcal{X})}$, which does not in general imply the condition \eqref{eq:M-condition}.

For the other direction, starting from condition \eqref{eq:M-condition} and applying the decreasing function $h$ to both sides of the inequality, we obtain
\begin{align*}
\mathcal{H}_{g,h}(\mu,(\mathcal{X},\nu))
&=h\left(\mathcal{M}_g\left(
\frac{1}{\mu(\mathcal X)}
\frac{d\mu}{d\nu},
(\mathcal X,\mu)
\right)\right) \\
&\leq
h\left( M_\mu^{-1}\left(\dfrac{1}{\mu(\mathcal{X}) }\dfrac{\dd \mu}{\dd \nu}, \mathcal{X} \right) \right)\\
&=h\left(\dfrac{1}{\nu(Y)} \right) \\
&\leq 
h\left(\dfrac{1}{\nu(\mathcal{X})} \right) \\
&= \mathcal{H}_{g,h}(\nu,(\mathcal{X},\nu))
\end{align*}
which is Postulate A. 
Equality throughout holds if and only if the normalized density is constant $\mu$-almost everywhere and $\nu(Y)=\nu(\mathcal X)$, which is equivalent to $\mu=c\nu$.

The fact that condition \eqref{eq:M-condition} is also satisfied by Shannon entropy and R\'enyi entropy of order $p \geq1$ is given by Lemma \ref{lem:M}. For $0<p<1$, the same inequality follows from
Lemma~\ref{lem:M-monotone}.
This result in the case of Shannon  and R\'enyi entropies also confirms that Postulate A is in fact the weakest condition in the hierarchy, in the sense that it is implied by the stronger additivity assumptions, but it does not imply them.

For (b), we treat the case in which $g$ is increasing; the case in which $g$ is
decreasing is identical after reversing the relevant inequalities, and we
indicate the necessary sign changes at the end.

Fix a finite measure space $(\mathcal{X},\mu)$, a positive measurable function
$f$, and write $\xi := \mu/\mu(\mathcal{X})$ for the associated probability
measure. Condition~\eqref{eq:M-condition}, restricted to this pair, reads
$$
  \mathcal{M}_g(f,(\mathcal{X},\mu)) \geq M^{-1}_\mu(f,\mathcal{X}).
$$
Since $g$ is increasing, applying $g$ to both sides preserves the
inequality:
\begin{equation}
\label{ineq:M-jensen}
  \int_\mathcal{X} g(f)\,d\xi \geq g\Big(M^{-1}_\mu(f,\mathcal{X})\Big)
  = g\left(\psi^{-1}\left(\int_\mathcal{X} \psi(f)\,d\xi\right)\right),
\end{equation}
where $\psi(x):=1/x$ is the generator of the harmonic mean. The map $\psi$
is an involution: $\psi^{-1}=\psi$. Let $\Phi:= g \circ \psi$ and $u := \psi(f) = 1/f$, so
that $f = \psi(u) = 1/u$ and $g(f) = g(1/u) = \Phi(u)$. Likewise,
$g(\psi^{-1}(\int \psi(f)\,d\xi)) = g(\psi(\int \psi(f)\,d\xi)) = \Phi(\int u\,d\xi)$. The inequality~\eqref{ineq:M-jensen} thus becomes
\begin{equation}
\label{eq:jensen-condition}
  \int_\mathcal{X} \Phi(u)\,d\xi \;\ge\; \Phi\Big(\int_\mathcal{X} u\,d\xi\Big),
\end{equation}
with equality in~\eqref{eq:M-condition} corresponding to equality in~\eqref{eq:jensen-condition}. As $f$
ranges over all positive measurable functions on all finite measure
spaces, so does $u=1/f$; hence condition~\eqref{eq:M-condition} holds universally if and
only if~\eqref{eq:jensen-condition} holds for every probability space $(\mathcal{X},\xi)$ and every
positive integrable $u$, that is, if and only if $\Phi$ satisfies Jensen's inequality unconditionally.

To see that ~\eqref{eq:jensen-condition} implies convexity of $\Phi$, we may take  $\mathcal{X}=\{1,2\}$,
$\xi = t\,\delta_1 + (1-t)\,\delta_2$ for an arbitrary $t\in(0,1)$, and let
$u$ take the values $a,b>0$ on the two atoms. Then~\eqref{eq:jensen-condition} reads
$$
  t\,\Phi(a) + (1-t)\,\Phi(b) \;\ge\; \Phi\big(ta+(1-t)b\big)
  \qquad \text{for all } t\in(0,1),\ a,b>0,
$$
which is precisely the statement that $\Phi$ is convex on $(0,\infty)$.
Since $a,b,t$ were arbitrary, $\Phi$ must be convex. 
The equality condition in~\eqref{eq:M-condition} shows that the inequality is strict
whenever the two values of $u$ are distinct. Thus $\Phi$ is strictly
convex. Conversely, strict Jensen inequality for $\Phi$ gives both the
mean comparison and its equality condition. 

If $g$ is decreasing, applying $g$ to $\mathcal{M}_g(f,(X,\mu))\ge M^{-1}_\mu(f,\mathcal{X})$
reverses the inequality~\eqref{ineq:M-jensen}, so~\eqref{eq:M-condition} becomes equivalent to the
\emph{reverse}-Jensen inequality $\int_{\mathcal{X}} \Phi(u)\,d\xi \leq
\Phi(\int_{\mathcal{X}} u\,d\xi)$ for all $\xi,u$, and the above
goes through verbatim with ``convex'' replaced by ``concave'' throughout.
\end{proof}

The comparison of generalized means used in the above proof is classical. In
particular, comparison of weighted quasi-arithmetic means can be expressed
through convexity or concavity of a composition of their generators; see Hardy, Littlewood, and P\'olya~\cite{HarLitPol52} and the weighted comparison results of Maksa and P\'ales~\cite{MaksaPales2010}, and Grabisch, Marichal, Mesiar, and Pap \cite{GrabischEtAl11}. The novelty of Theorem~\ref{thm:ent-admit} is not this comparison theorem itself---that is, the equivalence of conditions $(a)$ and $(b)$ in Theorem~\ref{thm:ent-admit}---but the identification of entropy admissibility under Postulate~A with universal domination of the harmonic mean. The classical mean-comparison criterion then converts that structural entropy requirement into a condition on the generator alone.

Generalized entropies based on Kolmogorov--Nagumo means have also been
studied in information theory and thermostatistics; see, for example,
\cite{AczelDaroczy75,CzaNau02,Duk10}. These works primarily examine entropy
constructions under additivity, pseudo-additivity, or thermodynamic
composition requirements. The present framework instead begins with
structural monotonicity relative to a reference measure and separates the
resulting restrictions on the mean generator from those imposed by the
choice of entropy scale and by product additivity.

\begin{lem}[Simple-density approximation]
Let $\mu$ be a nonzero finite measure, let $r>0$ $\mu$-almost
everywhere, and let $g:(0,\infty)\rightarrow \mathbb{R}$ be continuous and strictly
monotone. If $g(r)\in L^1(\mu)$, then there exist positive simple functions $r_n$, each bounded above
and bounded away from zero, such that
$$
g(r_n)\longrightarrow g(r)
\quad\text{in }L^1(\mu)
$$
and $\mu$-almost everywhere.
\end{lem}

\begin{proof}
Let $J:=g((0,\infty))$ and set $Y:=g(r)$. Choose compact intervals
$J_n\subset J$ increasing to $J$, and let $Y^{(n)}$ be the truncation of
$Y$ to $J_n$. Since $Y\in L^1(\mu)$, $Y^{(n)}\rightarrow Y$
in $L^1(\mu)$ and almost everywhere. Choose a simple function $Y_n$
taking values in $J_n$ such that
$$
|Y_n-Y^{(n)}|\leq\frac{1}{n}.
$$
Then set $r_n:=g^{-1}(Y_n)$.
Since $Y_n$ takes values in a compact subset of $J$, the function
$r_n$ is bounded above and bounded away from zero. Moreover,
$$
g(r_n)=Y_n\longrightarrow Y=g(r)
$$
in $L^1(\mu)$ and almost everywhere.
\end{proof}



\begin{proof}[Proof of Theorem~\ref{thm:ent-char}]

Continuity of $h$ follows from Postulate~2. We first note that its strict
decrease is forced by Postulate~A and the effective-support hypothesis.
Fix $0<r<s$, let $\mathcal X=\{1,2\}$, and define
$$
\nu(\{1\})=\frac1s,
\qquad
\nu(\{2\})=\frac1r-\frac1s.
$$
Let $Y=\{1\}$ and $\mu=\nu|_Y$. By the effective-support hypothesis,
$$
\mathcal H(\mu,(\mathcal X,\nu))
=
\mathcal H(\mu,(Y,\nu|_Y))
=
h(s).
$$
The second clause of Postulate~A gives
$$
\mathcal H(\nu,(\mathcal X,\nu))=h(r).
$$
Since $\mu$ is not proportional to $\nu$ on $\mathcal{X}$, the strict
equality condition in Postulate~A yields
$$
h(s)<h(r).
$$
Hence $h$ is strictly decreasing.

We first prove part~$(1)$. By the additional hypothesis of the theorem, it suffices to prove the representation on the restricted space. The $(g,h)$-entropy in Definition~\ref{def:gh-ent} is also unchanged under this restriction, since $\mu(Y)=\mu(\mathcal X)$, the Radon--Nikodym derivative restricts accordingly, and its defining integral is taken with respect to $\mu$.

On the restricted space one has $\mu|_Y\sim\nu|_Y$. We may therefore work on this space and, to simplify notation, continue to denote the restricted measures and space by $\mu$, $\nu$, and $\mathcal X$.


Let
$$
m:=\mu(\mathcal X)>0,
\qquad
\bar{r}
:=
\frac{1}{m}\frac{\dd\mu}{\dd\nu}.
$$
Since $\mu\sim\nu$ on the restricted space, $\bar r>0$ almost
everywhere and
$$
\dd\nu=\frac{1}{m\bar r}\,\dd\mu.
$$

We first prove the representation when $\bar r$ is simple, that is, for the positive simple-density structures appearing in the approximation hypothesis.
Write
$$
\bar r
=
\sum_{i=1}^n \bar r_i\mathbf 1_{A_i},
$$ 
where
$$
\mathcal X=A_1\sqcup\cdots\sqcup A_n,
\qquad
\bar r_i>0,
$$
After discarding any $\mu$-null cells, set
$$
m_i:=\mu(A_i),
\qquad
w_i:=\frac{m_i}{m},
\qquad
\mu_i:=\mu|_{A_i},
\qquad
\nu_i:=\frac{m}{m_i}\nu|_{A_i}.
$$
Then $w_i\nu_i=\nu|_{A_i}$,
and hence
$$
\bigoplus_{i=1}^n w_i\nu_i=\nu.
$$
Moreover,
$$
\frac{1}{m_i}\frac{\dd\mu_i}{\dd\nu_i}
=
\bar r_i.
$$
Indeed,
$$
\frac{\dd\mu_i}{\dd\nu_i}
=
\frac{m_i}{m}\frac{\dd\mu}{\dd\nu}
=
m_i\bar r_i.
$$
It follows that
$$
\mu_i=m_i\bar r_i\,\nu_i.
$$
Thus $\mu_i$ is proportional to $\nu_i$, and the equality case of
Postulate A gives
$$
\mathcal H(\mu_i,(A_i,\nu_i))
=
\mathcal H(\nu_i,(A_i,\nu_i)).
$$
Furthermore,
$$
\nu_i(A_i)
=
\frac{m}{m_i}\nu(A_i)
=
\frac{1}{\bar r_i}.
$$
By the definition of the induced atomic entropy scale,
$$
\mathcal H(\mu_i,(A_i,\nu_i))
=
h(\bar r_i).
$$

By Postulate 1, the original structure is identified with the corresponding
disjoint sum of the component structures. Repeated application of
Postulate $5'$ therefore gives
$$
\mathcal H(\mu,(\mathcal X,\nu))
=
\varphi^{-1}
\left(
\sum_{i=1}^n
w_i\varphi\bigl(h(\bar r_i)\bigr)
\right).
$$
Since $g=\varphi\circ h$, we have
$\varphi^{-1}=h\circ g^{-1}$,
and therefore
$$
\mathcal H(\mu,(\mathcal X,\nu))
=
h\circ g^{-1}
\left(
\sum_{i=1}^n w_i g(\bar r_i)
\right)
=
h\circ g^{-1}
\left(
\frac{1}{m}
\int_{\mathcal X}g(\bar r)\,\dd\mu
\right).
$$

Now let $\bar{r}$ be arbitrary, rather than necessarily simple.
Choose positive simple functions $\bar r_n$ as in the simple-density
approximation hypothesis, and define
$$
\dd\nu_n:=\frac{1}{m\bar r_n}\,\dd\mu.
$$
Then
$$
\frac{1}{m}\frac{\dd\mu}{\dd\nu_n}=\bar r_n,
$$
so the finite-range formula gives
$$
\mathcal H(\mu,(\mathcal X,\nu_n))
=
h\circ g^{-1}
\left(
\frac{1}{m}
\int_{\mathcal X}g(\bar r_n)\,\dd\mu
\right).
$$
Since
$$
g(\bar r_n)\longrightarrow g(\bar r)
\quad\text{in }L^1\left(\frac{\mu}{m}\right),
$$
we have
$$
\frac{1}{m}
\int_{\mathcal X}g(\bar r_n)\,\dd\mu
\longrightarrow
\frac{1}{m}
\int_{\mathcal X}g(\bar r)\,\dd\mu.
$$

Since $\dd\nu=\frac{1}{m\bar r}\,\dd\mu$, the reference measure $\nu$ is precisely $\nu_{\bar r}$ in the
notation of Definition~\ref{def:simple-density-continuity}.
The continuity of $h\circ g^{-1}$ and continuity under simple density
approximation now imply
$$
\mathcal H(\mu,(\mathcal X,\nu))
=
h\circ g^{-1}
\left(
\frac{1}{m}
\int_{\mathcal X}g(\bar r)\,\dd\mu
\right).
$$
Equivalently,
$$
\mathcal H(\mu,(\mathcal X,\nu))
=
h\left(
\mathcal M_g
\left(
\frac{1}{\mu(\mathcal X)}
\frac{\dd\mu}{\dd\nu},
(\mathcal X,\mu)
\right)
\right),
$$
which is the $(g,h)$-entropy in Definition~\ref{def:gh-ent}. This proves
part~$(1)$.

For part~$(2)$, Postulate $4_{\rm at}$ gives
$$
h(rs)=h(r)+h(s)
\qquad
\text{for all }r,s>0.
$$
By continuity, there is a constant $a$ such that
$$
h(r)=a\log r.
$$
Since $h$ is strictly decreasing, $a<0$. The uniform probability measure
on two atoms has constant normalized density $1/2$, and hence its entropy
is $h(1/2)$. Postulate 3 gives
$$
h(1/2)=\log 2.
$$
Consequently,
$$
a\log(1/2)=\log 2,
$$
so $a=-1$ and
$$
h(r)=-\log r.
$$
The formula in part~$(1)$ therefore becomes
$$
\mathcal H(\mu,(\mathcal X,\nu))
=
-\log
\mathcal M_g
\left(
\frac{1}{\mu(\mathcal X)}
\frac{\dd\mu}{\dd\nu},
(\mathcal X,\mu)
\right),
$$
which is the $g$-entropy in Definition~\ref{def:g-ent}. This proves
part~$(2)$.

For part~$(3)$, Postulate 4 implies Postulate $4_{\rm at}$ by applying
external additivity to one-atom structures. Hence part~$(2)$ gives
$h=-\log$, so $\mathcal H$ is a $g$-entropy.

If
$$
\bar r_i
=
\frac{1}{\mu_i(\mathcal X_i)}
\frac{\dd\mu_i}{\dd\nu_i},
\qquad
i=1,2,
$$
then the normalized density of the product structure is
$$
\bar r_1\otimes\bar r_2.
$$
External additivity therefore gives
$$
\mathcal M_g
\left(
\bar r_1\otimes\bar r_2,
(\mathcal X_1\times\mathcal X_2,\mu_1\otimes\mu_2)
\right)
=
\mathcal M_g(\bar r_1,(\mathcal X_1,\mu_1))
\mathcal M_g(\bar r_2,(\mathcal X_2,\mu_2)).
$$

By Lemma~\ref{lem:prod-space-mean} and
Remark~\ref{rem:multip-mean-char}, the continuous strictly monotone
generators having this multiplicative property are, up to affine changes
that do not alter the generalized mean,
$g(t)=\log t$ for $p=1$ or $g(t)=t^{p-1}$ for $p\ne1$.
These give the R\'enyi family, with the logarithmic case corresponding to
$p=1$. Postulate A, equivalently the admissibility criterion in
Theorem~\ref{thm:ent-admit}, restricts the order to $p>0$. This proves part~$(3)$.

For the first alternative in part~$(4)$,
Lemma~\ref{lem:4int-4} shows that Postulate $4_{\rm int}$ implies
Postulate 4. Part~$(3)$ therefore places the entropy in the R\'enyi
family. Lemma~\ref{lem:mean-factorization} then shows that the only member
of this family satisfying internal additivity is the one corresponding to
the geometric mean, namely $g(t)=\log t$ up to an affine change of scale. This is Shannon entropy.

For the second alternative, the fact that Postulates 4 and 5 characterize Shannon entropy is R\'enyi's classical theorem \cite{Ren61}. This proves part~$(4)$.
\end{proof}

In addition to the examples of entropies given in Sec.~\ref{subsec:admissible-examples}, other examples of admissible generators and their corresponding entropy functionals are listed in Table~\ref{tab:additional-generators}.

For compactness, we use the following notation in the table:
$$
m:=\mu(\mathcal X),
\qquad
R:=\frac{1}{m}\frac{\dd\mu}{\dd\nu},
\qquad
\mathbb E_\mu[F(R)]
:=
\frac{1}{m}\int_{\mathcal X}F(R)\,\dd\mu.
$$
Thus, for $h=-\log$,
$$
\mathcal H_g(\mu,(\mathcal X,\nu))
=
-\log g^{-1}\left(\mathbb E_\mu[g(R)]\right).
$$

\begin{table}[t]
\centering
\small
\setlength{\tabcolsep}{4pt}
\renewcommand{\arraystretch}{1.35}
\begin{tabularx}{\textwidth}{
@{}
>{\raggedright\arraybackslash}p{0.19\textwidth}
>{\raggedright\arraybackslash}p{0.27\textwidth}
>{\raggedright\arraybackslash}X
>{\raggedright\arraybackslash}p{0.18\textwidth}
@{}
}
\toprule
Generator
&
Resulting entropy
&
Connection
&
Comment
\\
\midrule

\textbf{Shifted logarithmic}
\par
$g_a(t)=\log(a+t)$,
\par
$a\geq 0$
&
$\displaystyle
\mathcal H_a
=
-\log\left[
\exp\left(
\mathbb E_\mu[\log(a+R)]
\right)-a
\right]$
&
For $a=0$, this gives Shannon entropy. After affine normalization,
the generator approaches the arithmetic-mean generator as
$a\to\infty$.
&
Interpolates between geometric and arithmetic generalized means.
\\[3pt]

\textbf{Reflected softplus}
\par
$\displaystyle
g_\lambda(t)
=
\log\left(1+e^{-\lambda/t}\right)$,
\par
$\lambda>0$
&
Let
$\displaystyle
m_\lambda
=
\mathbb E_\mu
\left[
\log\left(1+e^{-\lambda/R}\right)
\right].
$
Then
$\displaystyle
\mathcal H^{\rm sp}_\lambda
=
\log\left[
-\frac{1}{\lambda}
\log\left(e^{m_\lambda}-1\right)
\right].
$
&
Uses the softplus/log-sum-exp nonlinearity appearing in variational $f$-divergence methods \cite{NowCseTom16}, here applied to $R^{-1}$, the reciprocal of the normalized density ratio.
&
Bounded generator with an explicit inverse.
\\[3pt]

\textbf{Square root}
\par
$\displaystyle
g_\varepsilon(t)
=
\sqrt{t^{-2}+\varepsilon^2}-t^{-1}$,
\par
$\varepsilon>0$
&
Let
$\displaystyle
m_\varepsilon
=
\mathbb E_\mu
\left[
\sqrt{R^{-2}+\varepsilon^2}-R^{-1}
\right].
$
Then
$\displaystyle
\mathcal H_\varepsilon
=
\log\left(
\frac{\varepsilon^2-m_\varepsilon^2}
{2m_\varepsilon}
\right).
$
&
Its associated divergence generator is $\sqrt{1+\varepsilon^2t^2}-\sqrt{1+\varepsilon^2}$. Up to its normalization at $t=1$, this has the same square-root functional form as Charbonnier/pseudo-Huber robust losses \cite{Barron19}, with the
normalized density ratio $R$ as its argument.
&
Smooth algebraic generator with an explicit inverse.
\\[3pt]

\textbf{Reciprocal exponential}
\par
$\displaystyle
g_\lambda(t)=e^{-\lambda/t}$,
\par
$\lambda>0$
&
$\displaystyle
\mathcal H_\lambda
=
\log\left[
-\frac{1}{\lambda}
\log\left(
\mathbb E_\mu \left[e^{-\lambda/R} \right]
\right)
\right]
$
&
This is the single-scale case
$\Lambda=\delta_\lambda$ of the reciprocal Laplace-transform family.
&
Simplest explicit member of the kernel-generated class.
\\[3pt]

\textbf{Rational saturation}
\par
$\displaystyle
g(t)=\frac{t}{1+t}$
&
Let
$\displaystyle
m_{\rm rat}
=
\mathbb E_\mu
\left[
\frac{R}{1+R}
\right].
$
Then
$\displaystyle
\mathcal H_{\rm rat}
=
-\log\left(
\frac{m_{\rm rat}}{1-m_{\rm rat}}
\right).
$
&
The associated $f$-divergence is, up to an irrelevant linear term,
one quarter of the triangular discrimination.
&
Bounded rational generator with an explicit inverse.
\\

\bottomrule
\end{tabularx}
\caption{
Additional examples of admissible generators. All entropy functionals use the
logarithmic outer scale $h=-\log$ and are understood on their natural
domains.
}
\label{tab:additional-generators}
\end{table}

\begin{remark}[Entropy bounds from refining partitions]
Let $\mu\ll\nu$, set
$$
m:=\mu(\mathcal X)>0,
\qquad
r:=\frac{1}{m}\frac{d\mu}{d\nu},
$$
and let $\mathcal P$ be a finite measurable partition of $\mathcal X$. Denote by
$\sigma(\mathcal P)$ the $\sigma$-algebra generated by the cells of $\mathcal P$, and let
$$
r_{\mathcal P}
:=
\mathbb E_\nu
\left[
r\mid\sigma(\mathcal P)
\right].
$$
Thus $r_{\mathcal P}$ is constant on each cell $A\in\mathcal P$, and, whenever
$\nu(A)>0$,
$$
r_{\mathcal P}|_A
=
\frac{1}{\nu(A)}
\int_A r\,d\nu
=
\frac{\mu(A)}{m\nu(A)}.
$$
Let $\mu_{\mathcal P}$ be the measure satisfying
$$
\frac{d\mu_{\mathcal P}}{d\nu}
=
mr_{\mathcal P}.
$$
Equivalently, on each $A\in\mathcal P$ with $\nu(A)>0$,
$$
\mu_{\mathcal P}|_A
=
\frac{\mu(A)}{\nu(A)}\,\nu|_A.
$$
Thus $\mu_{\mathcal P}(A)=\mu(A)$, while the density is replaced on each
cell by its $\nu$-average.

Assume that the disjoint-sum composition law is increasing in each
entropy coordinate, as it is under Postulate $5'$. If $\mathcal Q$
refines $\mathcal P$, then $\mu_{\mathcal P}$ is obtained from
$\mu_{\mathcal Q}$ by replacing the measure on each cell of
$\mathcal P$ by the proportional-to-reference measure having the same
mass. Hence
$$
\mathcal H(\mu,(\mathcal X,\nu))
\leq
\mathcal H(\mu_{\mathcal Q},(\mathcal X,\nu))
\leq
\mathcal H(\mu_{\mathcal P},(\mathcal X,\nu)).
$$
For the one-cell partition $\mathcal P_0={\mathcal X}$,
$$
\mu_{\mathcal P_0}
=
\frac{\mu(\mathcal X)}{\nu(\mathcal X)}\,\nu,
$$
and hence
$$
\mathcal H(\mu_{\mathcal P_0},(\mathcal X,\nu)) =
\mathcal H(\nu,(\mathcal X,\nu)).
$$
Consequently,
$$
\mathcal H(\mu,(\mathcal X,\nu))
\leq
\mathcal H(\mu_{\mathcal P},(\mathcal X,\nu))
\leq
\mathcal H(\nu,(\mathcal X,\nu)).
$$
Let $(\mathcal F_n)$ be an increasing sequence of finite $\sigma$-algebras such that
$$
\sigma\left(\bigcup_{n\geq1}\mathcal F_n\right)=\sigma(r),
$$
where $\sigma(r)$ denotes the $\sigma$-algebra generated by $r$, and let
$\mathcal P_n$ be the finite partition associated with $\mathcal F_n$. Then
L\'evy's upward theorem gives
$$
\mathbb E_\nu[r\mid\mathcal F_n]
\longrightarrow
r
$$
$\nu$-almost everywhere and in $L^1(\nu)$. Hence, under continuity
of the entropy along this partition approximation,
$$
\mathcal H(\mu_{\mathcal P_n},(\mathcal X,\nu))
\downarrow
\mathcal H(\mu,(\mathcal X,\nu)).
$$
Thus, the entropy is recovered from decreasing finite-resolution upper
bounds, from the maximum reference entropy at the coarsest partition to
the entropy of the original Radon--Nikodym density in the limiting
refinement.
\end{remark}

\begin{remark}[The effective support]
The additional hypothesis in Theorem~\ref{thm:ent-char} specifies how
the entropy treats the part of the reference space on which the input
measure vanishes. If $\mu\ll\nu$ and $Y := \left\{ \frac{\dd\mu}{\dd\nu}>0 \right\}$,
then $\mu|_Y\sim\nu|_Y$. The hypothesis identifies the original reference-measure structure with its restriction to this effective support. In other words,  we identify the reference-measure structure $(\mu,(\mathcal X,\nu))$ with its effective restriction
$\bigl(\mu|_Y,(Y,\nu|_Y) \bigr)$.

This condition is satisfied automatically by the generalized
entropies considered in the theorem. Their defining means are taken
with respect to $\mu$, and therefore depend only on
$\dd\mu/\dd\nu$ on $Y$. The remaining reference mass on
$\mathcal X\setminus Y$ does not enter the defining integral.

The condition does not, however, follow from the other postulates
alone.\footnote{For example, for $c>0$ one may define
$$
\mathcal H_c(\mu,(\mathcal X,\nu))
:=
S_\nu(\mu,\mathcal X)
-
c\,\nu\left(
\left\{
\frac{\dd\mu}{\dd\nu}=0
\right\}
\right).
$$
The minus sign ensures that the additional term cannot increase the
entropy above the proportional-reference case required by
Postulate~A, and the term vanishes whenever $\mu$ is proportional to
$\nu$. The functional $\mathcal H_c$ satisfies Postulates A and 1--3, as well as the arithmetic mean-value property, but  is excluded upon restriction to the effective support.}
Such functionals also measure reference mass in regions lying outside
the information carried by the input measure $\mu$. The present
framework excludes this additional information and treats the entropy
as a functional of the behavior of $\mu$ relative to $\nu$ on the
support actually reached by $\mu$. Extensions retaining off-support
information may nevertheless be relevant in settings where absence
from the support itself carries information.

Restriction to the effective support is compatible with the two basic
operations used in the paper. If
$Y_i:=\{\dd\mu_i/\dd\nu_i>0\}$, then, up to the relevant null sets,
$$
Y_{\mu_1\otimes\mu_2,\nu_1\otimes\nu_2}
=
Y_1\times Y_2,
$$
while the effective support of the disjoint sum is
$Y_1\sqcup Y_2$. Thus restriction to the effective support commutes
with both products and disjoint sums.

Finally, the hypothesis does not discard reference mass universally.
When the input measure is the reference measure itself,
$\dd\nu/\dd\nu=1$ almost everywhere, so its effective support $Y=\mathcal X$  modulo $\nu$-null sets. In particular, the reference
entropy $\mathcal H(\nu,(\mathcal X,\nu))$ appearing in Postulate~A
continues to use the full reference structure.
\end{remark}

We end by showing how Postulate 5' itself can be derived from more fundamental properties, starting with   
\begin{itemize}
    \item[] \textsc{Postulate B.} (Intensivity) \quad If  $\mathcal{H}(\mu) = \mathcal{H}(\nu) = c$, then
    $\mathcal{H}(\mu\oplus \nu) = c$.
\end{itemize}

\begin{remark}
    We note, without proof, that while intensive entropies are ``mean-like,'' non-intensive entropies (\textit{e.g.}, those with  $\mathcal{H}(\mu_1\oplus \cdots \oplus \mu_n) = n^\alpha h$, for some $\alpha>0$)  would generally be ``volume-like,'' in the sense that their value grows with the number, or total size, of replicated components rather than remaining invariant under replication. Such functionals are therefore less suitable for measuring information of statistical structures in a scale-free manner.
\end{remark}

Embedded in R\'enyi's framework, as well as in ours, is the tacit underlying assumption that the entropy reflects only the information in the given structures of interest---in this case, the given measure spaces and the operations on them.
For the disjoint sum this means ``structural compatibility'' in the sense that the operation itself does not introduce new information not present in the measure spaces themselves. More precisely, we also assert
\begin{itemize}
    \item[] \textsc{Postulate C.} (Structural Compatibility) \quad 
Whenever
$\mu_i(\mathcal{X}_i)=\nu_i(\mathcal{Y}_i)$ and $\mathcal H(\mu_i)=\mathcal H(\nu_i)
$
for $i=1,2$, one has
$$
\mathcal H(\mu_1\oplus\mu_2)
=
\mathcal H(\nu_1\oplus\nu_2).
$$
\end{itemize}

\begin{prop}
\label{prop:derive-5'}
Assume Postulates B, C, and 1. 
Assume also that the range of $\mathcal H$ is an interval $I$ independent of the total mass, and that the induced disjoint-sum composition law, viewed as a function of the two positive masses and the two entropy coordinates, is continuous and strictly increasing in each entropy coordinate. 
Then Postulate $5'$ holds.
That is, there exists a continuous strictly monotone function
$\varphi:I\to\mathbb R$, unique up to an affine change of scale, such that
$$
\mathcal H(\mu\oplus\nu)
=
\varphi^{-1}
\left(
\frac{
\mu(\mathcal X)\varphi(\mathcal H(\mu))
+
\nu(\mathcal Y)\varphi(\mathcal H(\nu))
}{
\mu(\mathcal X)+\nu(\mathcal Y)
}
\right).
$$
\end{prop}

\begin{proof}[Proof of Proposition~\ref{prop:derive-5'}]
By Postulate C, the value of $\mathcal H(\mu \oplus \nu)$
depends only on the four quantities
$\mu(\mathcal X),
\nu(\mathcal Y)$,
$\mathcal H(\mu),
\mathcal H(\nu)$.
Thus disjoint sum induces a rule for combining two entropy values, weighted
by the masses of the two summands. More precisely, for $a,b>0$ and
$x,y\in I$, let
$$
F_{a,b}(x,y)
$$
denote the entropy of the disjoint sum of any two structures having masses
$a,b$ and entropies $x,y$, respectively. Postulate C implies that this
value is well defined.

Postulate 1 and commutativity of disjoint sum imply weighted symmetry.
Indeed, $\mu\oplus\nu\cong\nu\oplus\mu$ by canonical relabeling, and therefore
$$
F_{a,b}(x,y)=F_{b,a}(y,x).
$$

Postulate 1 and associativity of disjoint sum imply the grouping law. Since
$$
(\mu_1\oplus\mu_2)\oplus\mu_3
\cong
\mu_1\oplus(\mu_2\oplus\mu_3),
$$
we have
$$
F_{a+b,c}\bigl(F_{a,b}(x,y),z\bigr)
=
F_{a,b+c}\bigl(x,F_{b,c}(y,z)\bigr).
$$
Here the two successive applications of the disjoint-sum reference
convention give, under either parenthesization, the same reference measure
$$
\frac{a}{a+b+c}\nu_1
\oplus
\frac{b}{a+b+c}\nu_2
\oplus
\frac{c}{a+b+c}\nu_3.
$$
Thus the induced composition law is compatible with regrouping and
refinement.

Postulate B gives intensivity. Namely, $F_{a,b}(x,x)=x$.
Continuity and strict monotonicity in each entropy coordinate hold by
assumption.

Thus the family $\{F_{a,b}:a,b>0\}$ satisfies the continuity, strict
monotonicity, idempotence, symmetry, and grouping conditions in the
classical characterization of weighted quasi-arithmetic means
\cite{Kitagawa34}; see also \cite[p.~392]{Aczel48} for historical
context and \cite[Theorem~4]{GrabischEtAl11} for a modern fixed-weight
bisymmetry formulation. It follows that there exist a common continuous
strictly monotone function $\varphi:I\to\mathbb R$, unique up to affine change of scale, and a positive weight
scale $W:(0,\infty)\to(0,\infty)$ such that
$$
F_{a,b}(x,y)
=
\varphi^{-1}
\left(
\frac{W(a)\varphi(x)+W(b)\varphi(y)}
{W(a)+W(b)}
\right).
$$
The grouping law identifies a block of masses $a$ and $b$ with a single
block of mass $a+b$, and therefore gives
$$
W(a+b)=W(a)+W(b).
$$
Because $W$ is positive, it is strictly increasing: if $0<a<b$, then
$W(b)=W(a)+W(b-a)>W(a)$. Hence the additive function $W$ is linear,
so $W(a)=\kappa a$ for some $\kappa>0$. Cancelling $\kappa$ yields
$$
F_{a,b}(x,y)
=
\varphi^{-1}
\left(
\frac{a\varphi(x)+b\varphi(y)}
{a+b}
\right).
$$

Substituting
$
a=\mu(\mathcal{X})$, $b=\nu(\mathcal{Y})$, $
x=\mathcal{H}(\mu)$, $y=\mathcal{H}(\nu)$ gives Postulate $5'$.
\end{proof}

\vspace{1em}

\subsection*{Acknowledgments.} 
The author would like to thank Henry Cohn for his much appreciated feedback and mentorship.
The author used OpenAI's GPT-5.5 and GPT-5.6 models through ChatGPT for literature review, text generation, editing, and revision. The author independently reviewed and verified all resulting text, mathematical statements, substantive claims, and references and takes full responsibility for the manuscript.

 \bibliography{stat-ent-char}
 \bibliographystyle{alpha}
 \end{document}